\documentclass[10pt,journal,compsoc]{IEEEtran}
\ifCLASSOPTIONcompsoc
  \usepackage[nocompress]{cite}
\else
  \usepackage{cite}
\fi
\ifCLASSINFOpdf
\else
\fi

\usepackage{cite}
\usepackage{amsmath,amssymb,amsfonts}
\usepackage{graphicx}
\usepackage{textcomp}
\usepackage{xcolor}
\usepackage{subfiles}
\usepackage{amsthm}
\usepackage{algorithm}
\usepackage{algorithmicx}
\usepackage{algpseudocode}
\makeatletter
\let\MYcaption\@makecaption
\makeatother
\usepackage[font=footnotesize]{subcaption}
\makeatletter
\let\@makecaption\MYcaption
\makeatother

\theoremstyle{definition}
\newtheorem{definition}{Definition}
\newtheorem{theorem}{Theorem}
\newtheorem{lemma}{Lemma}

\usepackage{multicol}
\usepackage{multirow}
\usepackage{booktabs}

\usepackage{makecell}
\usepackage{url}

\usepackage{xcolor}

\newcommand\wzm[1]{\textcolor{black}{#1}}

\begin{document}
\title{Practical Vertical Federated Learning with Unsupervised Representation Learning}

\author{Zhaomin Wu\textsuperscript{\rm 1}, Qinbin Li\textsuperscript{\rm 2}, Bingsheng He\textsuperscript{\rm 3} \\
\{zhaomin\textsuperscript{\rm 1},qinbin\textsuperscript{\rm 2},hebs\textsuperscript{\rm 3}\}@comp.nus.edu.sg \\
National University of Singapore}%

\markboth{IEEE Transactions on Big Data, DOI: 10.1109/TBDATA.2022.3180117}%
{}

\IEEEoverridecommandlockouts

\IEEEpubid{\begin{minipage}{\textwidth}\ \\[12pt] \centering
  2332-7790 \copyright 2022 IEEE. Personal use is permitted, but republication/redistribution requires IEEE permission.\\
  See http://www.ieee.org/publications standards/publications/rights/index.html for more information.
\end{minipage}}

\IEEEtitleabstractindextext{%
\begin{abstract}
As societal concerns on data privacy recently increase, we have witnessed data silos among multiple parties in various applications. Federated learning emerges as a new learning paradigm that enables multiple parties to collaboratively train a machine learning model without sharing their raw data. \emph{Vertical federated learning}, where each party owns different features of the same set of samples and only a single party has the label, is an important and challenging topic in federated learning. Communication costs among different parties have been a major hurdle for practical vertical learning systems. In this paper, we propose a novel communication-efficient vertical federated learning algorithm named \textit{FedOnce}, which requires only one-shot communication among parties. To improve model accuracy and provide privacy guarantee, FedOnce features unsupervised learning representations in the federated setting and privacy-preserving techniques based on moments accountant. The comprehensive experiments on 10 datasets demonstrate that FedOnce achieves close performance compared to state-of-the-art vertical federated learning algorithms with much lower communication costs. Meanwhile, our privacy-preserving technique significantly outperforms the state-of-the-art approaches under the same privacy budget.
\end{abstract}

\begin{IEEEkeywords}
vertical federated learning, differential privacy
\end{IEEEkeywords}}

\maketitle

\IEEEpubidadjcol
\IEEEdisplaynontitleabstractindextext
\IEEEpeerreviewmaketitle

\IEEEraisesectionheading{\section{Introduction}}\label{sec:introduction}

\IEEEPARstart{D}{ata} silos, where physically distributed data cannot be collected on a central server due to privacy concerns and regulations, have been witnessed in many real-world applications. \wzm{ \textit{Federated learning} \cite{kairouz2019advances,li2019federated,zhang2022robust,zhou2019privacy}, which can be classified as \textit{horizontal federated learning} and \textit{vertical federated learning} according to the data partitioning among parties~\cite{yang2019federated}, is proposed as an emerging learning paradigm to enable collaborative training without exposing the data of each party. Different from horizontal federated learning where each party shares the same feature space but owns only a subset of samples, vertical federated learning, where each party shares the same sample space but only holds a subset of features, is a common and important scenario in practice.}

\textit{Privacy by design} (PbD), as a key concept in \wzm{General Data Protection Regulation (GDPR)} \cite{voigt2017eu}, states that personal data from different services/applications under the same company should not be directly shared with each other. This concept motivates our study to develop a collaborative machine learning framework to train a privacy-preserving model from all parties. Particularly, we consider each service/application as a party in federated learning and assume that the party holding labels, named \textit{host party}, is trusted by other parties without labels, named \textit{guest parties}. Such trust is practical as all the parties are under the same company. Fig.~\ref{fig:vfl} presents a concrete example of this setting\wzm{, which is an end-to-end procedure containing model training and model deployment}. An e-commerce company provides two major services: 1) an online shopping service managed by the retailing department and 2) an e-banking service managed by the Fintech department. This company aims to train a machine learning model to recommend its users the proper investing amount. Under the principle of privacy by design, although the Fintech department can be trusted by the retailing department, direct data sharing is prohibited. Therefore, these departments have to perform training by transferring intermediate information. After the training, a model can be released to provide recommendation services for the users of this e-commerce company. \wzm{Typical applications of federated learning includes Google Keyboard \cite{yang2018applied} and nationwide learning health system \cite{friedman2010achieving}.}

\begin{figure}
    \centering
    \includegraphics[width=.9\linewidth]{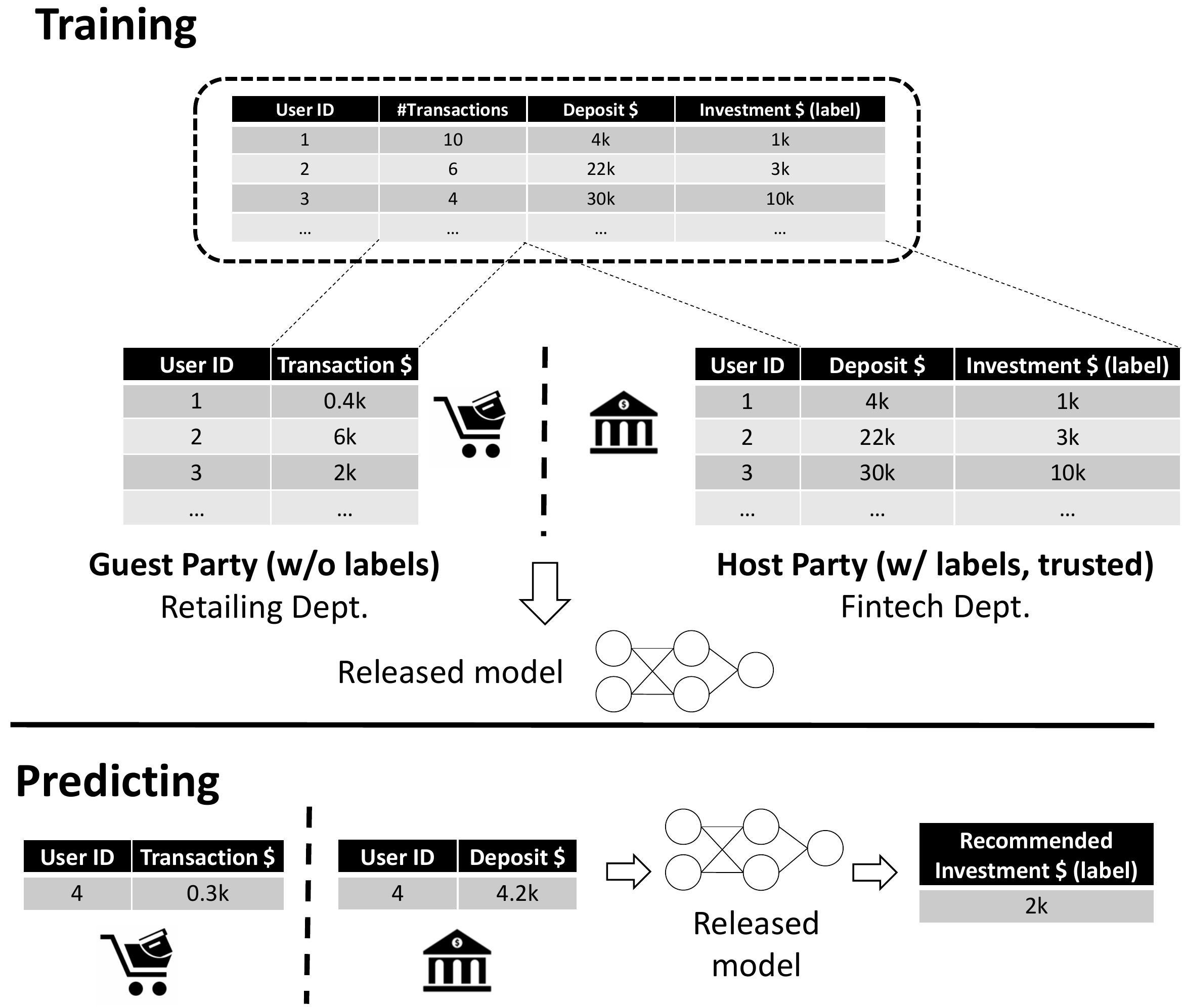}
    \caption{Example of vertical federated learning}
    \label{fig:vfl}
\end{figure}

\wzm{Considering the model training, }communication issues have limited the practical adoption of vertical federated learning in many real-world applications~\cite{loghin2020disruptions}. For example, such limitation becomes non-negligible when parties are 1) mobile devices communicating through Cellular networks, 2) IoT devices communicating through Wi-Fi networks, 3) seagoing ships communicating through satellites. These applications, featuring unstable connections and expensive communication cost, cannot be handled by existing vertical federated learning algorithms for two major reasons, including 1) \textbf{high stability requirement}: existing vertical federated learning approaches require all the parties to stay online or even synchronous during the entire training process, which is unrealistic for parties connected by unstable networks; 2) \textbf{high communication cost}: current studies in vertical federated learning usually incur large communication overhead, leading to the considerable economic cost. For example, in {each training iteration}, SecureBoost~\cite{cheng2019secureboost}, VF$^2$Boost~\cite{fu2021vf2boost}, and Pivot~\cite{wu2020privacy} exchange encrypted intermediate results (e,g, gradients) between the host party and guest parties. Similarly, when training {each batch} of data, SplitNN~\cite{vepakomma2018split} transfers outputs and gradients of a common layer between the host party and guest parties. All these approaches require batch/iteration-level synchronization and high communication cost.

To improve the practicability of vertical federated learning, designing a learning algorithm with a single communication round (i.e., one-shot), which can highly reduce the stability requirement and the communication cost, becomes increasingly desirable. Although some one-shot approaches \cite{guha2019one,zhou2020distilled,sharifnassab2021order} have been proposed for horizontal federated learning where each party shares the same feature space but owns only a subset of samples, these approaches cannot be directly adopted in vertical federated learning. The key challenge is that the labels only exist in a single party in vertical federated learning. These labels, which cannot be shared among parties due to privacy concerns, are necessary for the training of most models. For example, calculating gradients in neural networks or gradient boosting decision trees \cite{chen2016xgboost} relies on the labels. To the best of our knowledge, one-shot vertical federated learning algorithms have not been explored yet.

In this paper, we propose FedOnce, a novel one-shot vertical federated learning algorithm, the source code\footnote{https://github.com/JerryLife/FedOnce} of which is available online. FedOnce is designed for neural networks which are broadly used in many applications. Instead of aggregating the local models in horizontal learning, we propose the idea that the host party trains a model to aggregate the representations collected from guest parties. Specifically, guest parties first extract representative features by \emph{unsupervised learning}~\cite{bojanowski2017unsupervised}. These representative features, which can effectively boost the performance of federated learning, are sent to the host party. After receiving all the representative features, the host party trains an aggregation model based on these features and its own data. 

Considering the model deployment after the training, these released models are threatened by membership inference attack \cite{Shokri2017MembershipIA} \wzm{which detects the existence of training instances by inference. To protect the trained models in FedOnce from this attack, we exploit differential privacy \cite{dwork2014algorithmic} to theoretically guarantee that the training samples is indistinguishable given the trained model.} Despite many studies that apply differential privacy to horizontal federated learning \cite{zhao2018inprivate,jayaraman2018distributed,liu2020flame}, their privacy analyses cannot be directly applied to vertical federated learning. The main reason is that vertical federated learning requires analyzing inter-party privacy loss since each data record is distributed to multiple parties, whereas horizontal federated learning only requires analyzing inner-party privacy loss since each data record is held by a single party. Some existing studies~\cite{lou2018uplink,yao2019privacy,lou2020uplink} employ differential privacy in vertical federated learning. Nonetheless, these studies, focusing on inner-party privacy loss, only consider simple composition when analyzing inter-party privacy loss. In FedOnce, by refining the analysis of inter-party privacy loss based on moments accountant~\cite{abadi2016deep}, we significantly reduce the privacy loss compared to existing approaches.

Our main contribution can be summarized as:
\begin{itemize}
    \item We propose a practical one-shot vertical federated learning algorithm \textit{FedOnce} by effectively taking advantage of unsupervised learning.  
    \item We develop privacy techniques for FedOnce under the notion of differential privacy. Meanwhile, we prove a tighter bound on the privacy loss across parties compared to the current work on vertical federated learning.
    \item Our experiments demonstrate that FedOnce reduces the required communication cost to achieve competitive performance against state-of-the-art vertical federated learning algorithms \cite{vepakomma2018split,cheng2019secureboost} by at most 94\%. Meanwhile, our privacy-preserving technique significantly outperforms other approaches under the same privacy budget.
\end{itemize}
 
\section{Preliminaries}\label{sec:preliminaries}
\subsection{Differential Privacy}
Differential privacy \cite{dwork2008differential} is a technique that provides a theoretical guarantee of the privacy of individuals.

\begin{definition} (Differential Privacy) Suppose \(\mathcal{M}:\mathcal{D}\rightarrow\mathcal{R}\) is a randomized mechanism, $D$ and $D'$ are two neighboring datasets in domain $\mathcal{D}$, i.e., $\|D-D'\|\leq 1$. Then, $\mathcal{M}$ satisfies $(\varepsilon,\delta)$-differential privacy if for any outputs $S\subseteq \mathcal{R}$, 

$$ \Pr[\mathcal{M}(D)\in S]\leq e^\varepsilon \Pr[\mathcal{M}(D')\in S]+\delta$$
where $\varepsilon$ and $\delta$ are privacy parameters.
\end{definition}
In the definition of $(\varepsilon,\delta)$-differential privacy, $\varepsilon$ bounds the difference between the outputs of two neighboring datasets, and this bound might be violated by a small probability $\delta$, which is typically smaller than $1/|D|$, where $|D|$ is the number of samples in $D$.

Among many privacy mechanisms that achieve differential privacy, \textit{Gaussian Mechanism}~\cite{dwork2014algorithmic} is widely adopted in deep learning. It ensures the differential privacy of a deterministic function $f$ by adding Gaussian noise to the output. The noise scale $\sigma$ is determined by the privacy parameters $\varepsilon,\delta$ and the $\ell_2$-\textit{sensitivity} of the function.
\begin{definition} ($\ell_2$-sensitivity)
The $\ell_2$-sensitivity of function $f:\mathcal{D}\rightarrow\mathcal{R}$ is defined as
$$ \Delta_2 (f) \triangleq \max_{D,D'}\Vert f(D)-f(D')\Vert_2 $$
where $D$, $D'$ are neighboring datasets, i.e., $\|D-D'\|\leq 1$.
\end{definition}
\begin{theorem} (Gaussian Mechanism)
For any $\varepsilon\in(0,1)$, $c^2>2\ln(1.25/\delta),\;\sigma=c\Delta_2(f)/\varepsilon$, Gaussian mechanism
$$ \mathcal{M}(D)\triangleq f(D)+\mathcal{N}\left(0,\sigma^2\textbf{I} \right) $$
satisfies $(\varepsilon,\delta)$-differential privacy.
\end{theorem}
Gaussian mechanism achieves the differential privacy of a single function. Furthermore, if multiple mechanisms are applied to one dataset, the overall differential privacy follows simple composition~\cite{dwork2014algorithmic}. 
\begin{theorem} (Simple Composition)\label{thm:simple}
Let $\mathcal{M}_i:\mathcal{D}\rightarrow \mathcal{R}$ be a $(\varepsilon_i,\delta_i)$-differential privacy mechanism. Then a composition of $k$ such mechanisms $\mathcal{M}_k=(\mathcal{M}_1,...,\mathcal{M}_k)$ satisfies $\left(\sum_{i=1}^k{\varepsilon_i}, \sum_{i=1}^k{\delta_i}\right)$-differential privacy.
\end{theorem}

Simple composition holds regardless of the distribution of $\mathcal{M}_i$ as long as it satisfies $(\varepsilon_i,\delta_i)$-differential privacy. When all the mechanisms $\mathcal{M}_i$ are Gaussian mechanisms, this privacy bound can be improved by moments accountant~\cite{abadi2016deep} which first bounds the Rényi divergence~\cite{mironov2017renyi} of the composition and then calculates the overall privacy loss $\varepsilon$ given a specific $\delta$. The main theorem of moments accountant is concluded as Theorem~\ref{thm:ma}.
\begin{theorem}\label{thm:ma}
Given the sampling probability of SGD $q$, the number of epochs $T$, there exists constants $c_1$ and $c_2$ so that for $\forall\varepsilon<c_1q^2 T,\;\forall\delta>0$, \wzm{differentially private SGD \cite{abadi2016deep} with noise scale $\sigma$} is $(\varepsilon,\delta)$-differential privacy if we choose
$$ \sigma\geq c_2\frac{q\sqrt{log(1/\delta)T}}{\varepsilon} $$
\end{theorem}

The proof of moments accountant relies on the definition of privacy loss and $\lambda$-th moments. Auxiliary input (e.g. hyperparameters) in the original definition is omitted for simplicity.

\begin{definition}\label{def:privacy_loss} (Privacy Loss)
For two neighboring datasets $D$ and $D'$, an output $o$ and a mechanism $\mathcal{M}$, the privacy loss is defined as
$$ c(o;\mathcal{M},D,D')\triangleq \ln\frac{\Pr[\mathcal{M}(D)=o]}{\Pr[\mathcal{M}(D')=o]} $$
\end{definition}

\begin{definition}\label{def:moments} ($\lambda$-th moments)
The $\lambda$-th moments is defined as the log of the moments generating function of privacy loss evaluated at value $\lambda$:
$$ \alpha_\mathcal{M}(\lambda;D,D')\triangleq \ln \mathbb{E}_{o\sim\mathcal{M}(D)}[\exp{\lambda c(o;D,D')}] $$
Furthermore, the maximum of all possible $\lambda$-th moments is defined as
$$ \alpha_\mathcal{M}(\lambda)\triangleq\max_{D,D'} \alpha_\mathcal{M}(\lambda;D,D') $$
\end{definition}

In \cite{abadi2016deep}, the authors first prove the moments bound $\alpha_\mathcal{M}(\lambda)$ for a single Gaussian mechanism (Lemma~\ref{lem:gaussian}). Secondly, they prove how the bound accumulates in a sequence of Gaussian mechanisms (Lemma~\ref{lem:ori_compose}). Finally, the relationship between this moments bound and differential privacy parameters are proved according to Lemma~\ref{lem:tail}.

\begin{lemma}\label{lem:gaussian}
Suppose that $f:D\rightarrow\mathbb{R}^p$ with $\|f(\cdot)\|_2\leq 1$. Suppose $\sigma\geq 1$ and $J$ is a sample from $[n]$ where each $i\in[n]$ is chosen independently with probability $q\leq \frac{1}{16\sigma}$. Then, for any positive integer $\lambda\leq \sigma^2\ln \frac{1}{q\sigma}$, the Gaussian mechanism $\mathcal{M}(d)=\sum_{i\in J}f(d_i)+\mathcal{N}(0,\sigma^2\textbf{I})$ satisfies
$$ \alpha_\mathcal{M}(\lambda)\leq \frac{q^2\lambda(\lambda+1)}{(1-q)\sigma^2}+O(q^3\lambda^3/\sigma^3) $$
\end{lemma}

\begin{lemma}\label{lem:ori_compose}
Suppose $\mathcal{M}$ contains a sequence of $T$ adaptive mechanisms $\mathcal{M}_1,...,\mathcal{M}_T$, where $\mathcal{M}_i:\prod_{j=1}^{i-1}\mathcal{R}_j\times\mathcal{D}\rightarrow\mathcal{R}_i$. For $\forall\;\lambda>0$, we have
$$ \alpha_\mathcal{M}(\lambda)\leq \sum_{i=1}^T\alpha_{\mathcal{M}_i}(\lambda) $$
\end{lemma}

\begin{lemma}\label{lem:tail}
For any $\varepsilon>0$, the mechanism $\mathcal{M}$ is $(\varepsilon,\delta)$-differential private if
$$ \delta=\min_\lambda\exp(\alpha_\mathcal{M}(\lambda)-\lambda\varepsilon) $$
\end{lemma}

\subsection{Unsupervised Learning}
Unsupervised learning trains a model without labels. With this model, each sample can be tagged as a feature map, named the \textit{representation} of this sample.

FedOnce does not rely on specific unsupervised learning algorithms. In this paper, we choose NAT~\cite{bojanowski2017unsupervised}, which is a simple and effective unsupervised learning algorithm suitable for both image datasets and multivariate datasets. In NAT, in order to learn a representation $R:n\times d$ of $n$ samples with $d$ dimensions, a random representation $C:n\times d$ is generated. Denoting the neural network model to generate $R$ as $f_\theta(\cdot)$, the goal of NAT can be expressed as the Equation~\ref{eq:natloss}.

\begin{equation}\label{eq:natloss}
    \min_{P,\theta}{\frac{1}{2n}\left\| f_{\theta}(D)-PC\right\|_F^2} 
\end{equation}
where $P$ is a permutation matrix, D is the dataset without labels, and $PC$ (i.e., the matrix multiplication of $P$ and $C$) is a permutation of $C$. 

In each iteration of training, $\theta$ is then optimized by gradient descent. For every $t_{P}$ iterations, $P$ is optimized by Hungarian algorithm~\cite{kuhn1955hungarian}. The updating frequency $t_{P}$ can be adjusted as a hyperparameter. After the training process, the representation can be obtained by predicting with model $\theta$, i.e., $R=f_\theta(D)$.

\subsection{SplitNN}
SplitNN \cite{vepakomma2018split} is a simple and effective vertical federated learning\footnote{Though some studies categorize SplitNN as split learning, we denote SplitNN as a vertical federated learning algorithm for simplicity.} model based on neural network. As explained in Fig.~\ref{fig:splitnn}, models are split into multiple parties. Each guest party holds a local model; the host party holds a local model and an aggregation model. In forward propagation, first, all the parties forward propagate independently using their local models. Then, the outputs of local models on guest parties are sent to the host party. The host party concatenates the outputs from all the parties (including itself) and feeds the concatenated output into the aggregation model $\theta_{agg}$. The aggregation model continues forward propagating to produce the final prediction. In the backpropagation, first, the loss and gradients are calculated based on the final prediction. Then, the gradients w.r.t. to $\theta_1, \theta_2, \theta_3$ are calculated by backpropagating through $\theta_{agg}$. The gradients w.r.t. $\theta_2, \theta_3$ are sent to two guest parties, respectively. Finally, all the parties perform backpropagation on their local models. In each iteration of training, both forward propagation and backpropagation are performed for one time. Therefore, the training incurs two communication rounds between the host party and each guest party every iteration, which is impractical when parties are connected by unstable and expensive networks. This pitfall of SplitNN motivates our design of a one-shot vertical federated learning framework. 
 
\begin{figure}
    \centering
    \includegraphics[width=0.9\linewidth]{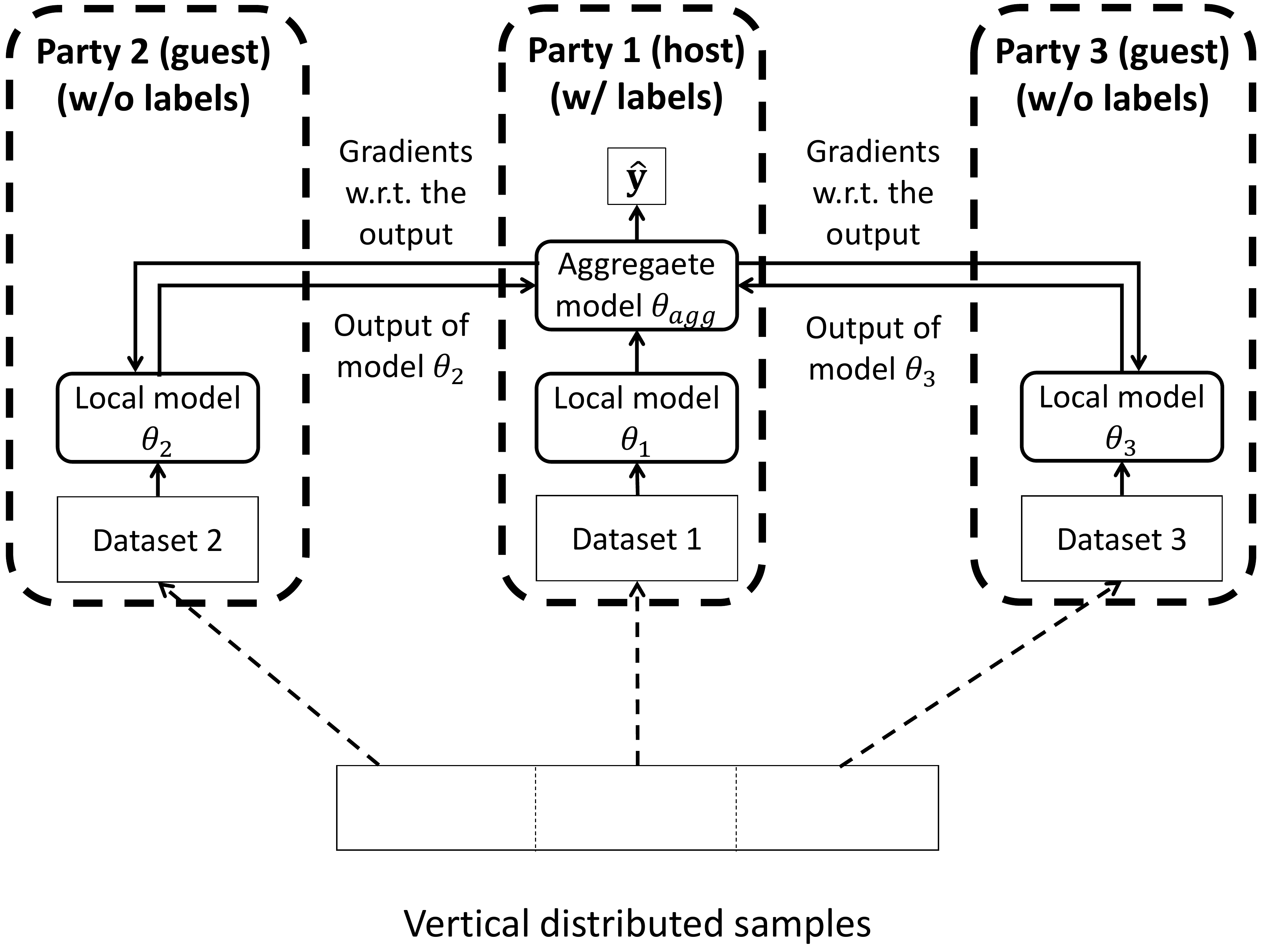}
    \caption{Structure of SplitNN}
    \label{fig:splitnn}
\end{figure}

\section{The FedOnce Framework}\label{sec:framework}
In this section, we propose the framework of FedOnce. FedOnce with no protection on released models is denoted as FedOnce-L0, corresponding to the scenario where the trained model is used inside the company, thus free from privacy attacks. FedOnce ensuring differential privacy of the released model is denoted as FedOnce-L1, corresponding to the scenario where the trained model is released to the public (either as white-box or black-box), thus threatened by the membership inference attack~\cite{Shokri2017MembershipIA}. Our objective is to solve the problem stated in section~\ref{subsec:problem} by one-shot communication. The framework is presented in Figure~\ref{fig:framework}.

\subsection{Problem Formulation}\label{subsec:problem}
\noindent\textbf{FedOnce-L0} Given a global dataset $\{\mathbf{x}_i, y_i\}_{i=1}^n$, the features of each sample $\mathbf{x}_i$ is distributed among $k$ parties (i.e., $\mathbf{x}_i = [\mathbf{x}_i^1, \mathbf{x}_i^2, ..., \mathbf{x}_i^k]$) while only one party has the labels $\mathbf{y}=\{y_i\}_{i=1}^n$. Without loss of generality, we assume party 1 is the host party and the other parties are guest parties. We denote the local dataset on party $j$ as $\mathbf{x}^j\triangleq \{\mathbf{x}_i^j\}_{i=1}^n$. All the parties need to collaboratively train a model $f(\theta;\mathbf{x}^1,\mathbf{x}^2,...,\mathbf{x}^k)$ without exchanging their feature vectors or labels with each other. \textbf{One communication round with the host party is permitted for each guest party.} The training process can be formulated as
$$ \min_\theta \frac{1}{n}\sum_{i=1}^{n}L(f(\theta;\mathbf{x}^1_i,\mathbf{x}^2_i,...,\mathbf{x}^k_i);y_i)+\lambda\omega(\theta) $$
where $L(f(\theta;\mathbf{x}^1_i,\mathbf{x}^2_i,...,\mathbf{x}^k_i);y_i)$ is the loss function and $\lambda\omega(\theta)$ is the regularization term.

Note that in the setting of FedOnce-L0, we assume that the samples have already been aligned among the parties (i.e., party $j$ knows $\mathbf{x}_i^j$ contains partial features of $\mathbf{x}_i$). Such alignment can be implemented by privacy-preserving record linkage \cite{vidanage2019efficient}, which is an orthogonal topic to this paper. This is a common setting in vertical federated learning \cite{hu2019fdml,cheng2019secureboost,fu2021vf2boost,wu2020privacy}. 

\noindent\textbf{FedOnce-L1} In FedOnce-L1, we focus on the privacy risk of model deployment after the training. We assume the host party is trusted by guest parties. The attackers can obtain host and guest models including all the parameters required for inference. The goal of FedOnce-L1 is to protect the released models against membership inference attack \cite{nasr2019comprehensive} by ensuring ($\varepsilon,\delta$)-differential privacy of the released model. The details on privacy design are presented in Section~\ref{sec:privacy}.

\begin{figure}[htpb]
\centering
  \includegraphics[width=0.95\linewidth]{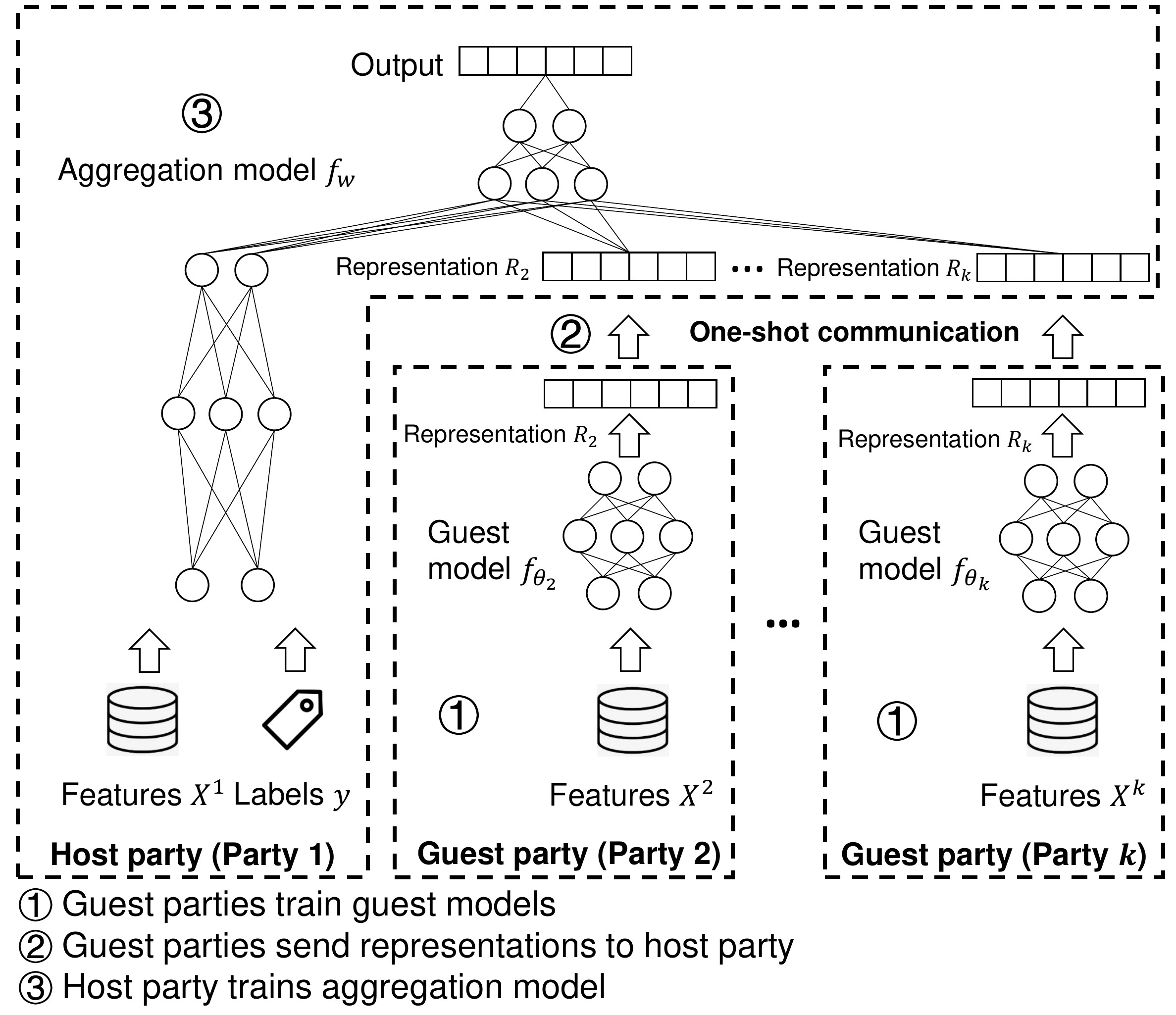}
  \caption{Framework of FedOnce}
  \label{fig:framework}
\end{figure}

\subsection{Training}
Suppose among $k$ parties $\mathcal{P}_1, \mathcal{P}_2, ..., \mathcal{P}_k$, $\mathcal{P}_1$ is the host party and $\mathcal{P}_2, ..., \mathcal{P}_k$ are guest parties. In the training process, guest parties first extract representations of their data by unsupervised learning based on \cite{bojanowski2017unsupervised}. The representations, reflecting the high-level features of data on guest parties, are then sent to the host party, which is the only communication required in FedOnce. Finally, the host party collects all the representations and trains an aggregation model together with its own data. The training process of FedOnce is summarized as Algorithm~\ref{alg:train}, where the regularization term is omitted for simplicity. For convenience, we denote $R_i,...,R_j$ as $R_{i:j}$. Although the training is performed by batch in practice, we present the training procedure by epoch in Algorithm~\ref{alg:train} for simplicity.

Specifically, guest parties first initialize random targets $C_j$ (line 3). In each epoch $t$, Hungarian algorithm~\cite{kuhn1955hungarian} is adopted to find the best permutation of $C_j$, denoted as $\tilde{P}^{(t)}$, that minimizes the loss (line 5). Because of the high computational complexity, the Hungarian algorithm finds the best permutation in each batch rather than in the whole dataset. Later, this permutation is used to calculate the gradient w.r.t. parameters $\theta_j^{(t)}$ (line 6-7). If the released model requires no protection, i.e., FedOnce-L0, $\theta_j$ is directly updated by the gradients (line 9). Otherwise, in FedOnce-L1, the gradients are clipped to bound the sensitivity (line 11). Gaussian noise is then generated and added to the gradients (line 12). Finally, $\theta_j^{(t)}$ is updated by these gradients with noise (line 13). After each guest party $\mathcal{P}_j(j\in[2,k])$ completes training, the calculated representation $R_j$ is sent to the host party (line 16). 

Once receiving the representations from all guest parties, the host party starts to train an aggregation model $f_w$, which takes representations $R_{2:j}$ and its own data $x_1$ as inputs and outputs the final prediction. First, gradients w.r.t. $w$ are calculated (lines 21-22). In FedOnce-L0, the gradients are directly utilized to update $w$ (line 24). Otherwise, in FedOnce-L1, gradients are clipped (line 26); Gaussian noise is generated and added to the gradients (lines 26-27). Finally, model parameters $w$ are updated by these gradients with noise (line 28).

\noindent\textbf{Extension to Multi-Round Training.} While achieving impressive performance in one communication round, FedOnce can be further improved if more communication rounds are permitted. At the end of the training process of FedOnce, each party holds a guest model and the host party additionally holds an aggregation model; the combination of these models can be regarded as a pre-trained model of SplitNN. Therefore, the training process of SplitNN can be directly applied to further boost the performance of FedOnce. Specifically, in each iteration, guest parties perform forward propagation and transfer intermediate outputs to the host party. The host party receives the outputs and continues forward propagation in the aggregation model to obtain predictions $\hat{\mathbf{y}}$. After calculating the loss with $\hat{\mathbf{y}}$ the labels $\mathbf{y}$, the host party performs backpropagation and transfers gradients to guest parties. Finally, the received gradients are used to perform backpropagation to update guest models. We denote this extension as \textit{multi-round FedOnce-L0}.

\algnewcommand{\LineComment}[1]{\State \(\triangleright\) #1}

\begin{algorithm}[t!]
\caption{Training Process of FedOnce}\label{alg:train}
\begin{algorithmic}[1]
\Require training data $\left(\mathbf{x}^{1},\mathbf{x}^{2},...,\mathbf{x}^{k}\right)$ distributed on $k$ parties; labels $\mathbf{y}$ on host party; number of epochs $T_j$ of each party $j$; number of samples $n$
\LineComment{Guest Parties}
\For{party $j\in[2,k]$} \Comment{Parallel}
\State \textbf{Initialize} guest model $f_{\theta_j}$, random targets $C_j$
\For{epoch $t\in[T_j]$}
\State $\tilde{P}^{(t)}\leftarrow arg\min_{P^{(t)}}\frac{1}{2n}\| f^{(t)}_{\theta_j}(\mathbf{x}^{j})-P^{(t)}C_j\|_F^2$
\State $\tilde{L}^{(t)}(\theta_j,\mathbf{x}^{j})\leftarrow \frac{1}{2n}\| f^{(t)}_{\theta_j}(\mathbf{x}^{j})-\tilde{P}^{(t)}C_j\|_F^2$
\State $\textbf{g}^{(t)}\leftarrow\nabla_{\theta_j}L^{(t)}(\theta_j,\mathbf{x}^{j}) $
\If{FedOnce-L0}
\State $\theta_j^{(t+1)}=\theta_j^{(t)}-\eta^{(t)}\textbf{g}^{(t)}$
\ElsIf{FedOnce-L1}
\State $\textbf{g}^{(t)}_{clip}=\textbf{g}^{(t)}/max(1,\frac{\|\textbf{g}^{(t)}\|_2}{\Omega})$
\State $\textbf{g}^{(t)}_{noise}=\textbf{g}^{(t)}_{clip}+\mathcal{N}(0,\sigma^2\Omega^2\textbf{I})$
\State $\theta_j^{(t+1)}=\theta_j^{(t)}-\eta^{(t)}\textbf{g}^{(t)}_{noise}$
\EndIf
\EndFor
\State Send representation $R_j=f_{\theta_j}(\mathbf{x}^{j})$ to host party
\EndFor
\LineComment{Host Party}
\State Collect $R_j (j\in[2,k])$ from all guest parties
\For{epoch $t\in[T_j]$}
\State $L^{(t)}(w,\mathbf{x}^{1},R_{2:j})\leftarrow \frac{1}{2n}\| f^{(t)}_{w}(\mathbf{x}^{1},R_{2:j})-\mathbf{y}\|_F^2$
\State $\textbf{g}^{(t)}\leftarrow\nabla_{w}L^{(t)}(w,\mathbf{x}^{1},R_{2:j}) $
\If{FedOnce-L0}
\State $w^{(t+1)}=w^{(t)}-\eta^{(t)}\textbf{g}^{(t)}$
\ElsIf{FedOnce-L1}
\State $\textbf{g}^{(t)}_{clip}=\textbf{g}^{(t)}/max(1,\frac{\|\textbf{g}^{(t)}\|_2}{\Omega})$
\State $\textbf{g}^{(t)}_{noise}=\textbf{g}^{(t)}_{clip}+\mathcal{N}(0,\sigma^2\Omega^2\textbf{I})$
\State $w^{(t+1)}=w^{(t)}-\eta^{(t)}\textbf{g}^{(t)}_{noise}$
\EndIf
\EndFor
\State \Return model parameters $\theta_j(j\in[2,k]), w$
\end{algorithmic}
\end{algorithm}

\subsection{Prediction}
The prediction is conducted by all the parties collaboratively according to Algorithm~\ref{alg:pred}. The features of testing data are distributed on $k$ parties, i.e., $\mathbf{x}=\left(\mathbf{x}^{1}, \mathbf{x}^{2}, ..., \mathbf{x}^{k}\right)$. First, each guest party $\mathcal{P}_j$ makes a prediction according to its guest model $f_{\theta_j}$ (line 3). These predicted representations $R_j$ are then sent to the host party (line 4) which makes the final prediction based on $f_{\theta_j}(\mathbf{x}_i^{j})$ and partial sample $\mathbf{x}_i^{1}$ (line 7).

\begin{algorithm}
\caption{Predicting Process of FedOnce}\label{alg:pred}
\begin{algorithmic}[1]
\Require testing data $\left(\mathbf{x}^{1},\mathbf{x}^{2},...,\mathbf{x}^{k}\right)$ distributed on $k$ parties; trained models $f_{\theta_j}(\cdot)\;(j\in[1,k]),f_w$;
\LineComment{Guest parties}
\For{party $j\in[2,k]$}
\State $R_j=f_{\theta_j}(\mathbf{x}^{j})$
\State Send $R_j$ to host party
\EndFor
\LineComment{Host party}
\State $\hat{\mathbf{y}} = f_w(\mathbf{x}^{1},R_{2:j})$
\State \Return prediction $\hat{\mathbf{y}}$
\end{algorithmic}
\end{algorithm}

\section{Privacy}\label{sec:privacy}
\subsection{Privacy of Released Model}

FedOnce-L1 adopts differential privacy to ensure each sample cannot be distinguished through the released model. The noise addition method largely follows \cite{abadi2016deep}. As demonstrated in Algorithm~\ref{alg:train}, gradients are first clipped by a threshold $\Omega$ in each iteration of the training process to ensure that $\ell_2$-sensitivity of all the gradients are bounded by $\Omega$, i.e., $\|\textbf{g}_{clip}\|_2\leq\Omega$. After gradient clipping, independent Gaussian noise is added to each gradient.

\wzm{Compared to the original moments accountant, where the privacy loss is only accumulated across iterations, the privacy loss in differentially private federated learning also accumulates across parties.} Regarding the division of overall privacy budget $\varepsilon$ to each party, existing studies~\cite{lou2018uplink,yao2019privacy} on differentially private vertical federated learning simply divide $\varepsilon$ equally to each party. This method, denoted as \textit{simple division}, satisfies differential privacy according to simple composition (Theorem~\ref{thm:simple}), which is over-conservative and leads to a small privacy budget for each party. Hence, we propose \textit{moments division}, which is a natural extension of moments accountant. Since inner-party privacy loss is accumulated according to moments accountant, the Rényi divergence~\cite{mironov2017renyi} for each party is calculated along with the privacy loss. Instead of accumulating $\varepsilon$, we accumulate the Rényi divergences of parties and calculate the overall $\varepsilon$ given $\delta$. Our main result is summarized as Theorem~\ref{thm:main}.

\begin{theorem}\label{thm:main}
Given the sampling probability of SGD $q$, the number of epochs $T_j$ of each party $\mathcal{P}_j$, and the number of parties $k$, there exists constants $c_1$ and $c_2$ so that for $\forall\varepsilon<c_1q^2\sum_{j=1}^k T_j,\;\forall\delta>0$, FedOnce-L1 is $(\varepsilon,\delta)$-differential privacy if we choose
$$ \sigma\geq c_2\frac{q\sqrt{log(1/\delta)\sum_{j=1}^k{T_j}}}{\varepsilon} $$
\end{theorem}

The proof of Theorem~\ref{thm:main} is provided as follows. The major difference between our proof of FedOnce-L1 and moments accountant lies in the accumulation of moments, where moments accountant only accumulates across iterations, but FedOnce-L1 also accumulates across parties. Specifically, we extend Lemma~\ref{lem:ori_compose} to the setting of FedOnce-L1 as Lemma~\ref{lem:compose}. 

\begin{lemma}\label{lem:compose}
Suppose $\mathcal{M}$ contains $k$ sequences of Gaussian mechanisms $\mathcal{M}^1,...,\mathcal{M}^k$, each sequence $j$ contains $T$ Gaussian mechanisms $\mathcal{M}^j_1,...,\mathcal{M}^j_T$. For $j\in[1,k-1]$, $\mathcal{M}_i^j:\mathcal{D}\rightarrow\mathcal{R}_j$ takes original dataset as input (guest parties). For $j=k$, $\mathcal{M}_i^k:\prod_{j=1}^{k-1}\mathcal{R}_j\times\mathcal{D}\rightarrow\mathcal{R}_k$ take the orignal dataset and the output of other mechanisms as input (host party). For $\forall\;\lambda>0$, we have
$$ \alpha_\mathcal{M}(\lambda)\leq \sum_{j=1}^k\sum_{i=1}^T\alpha_{\mathcal{M}_i^j}(\lambda) $$
\end{lemma}

\begin{proof}
Intuitively, the sequences of mechanisms on all the parties can be regarded as a large sequence of adaptive mechanisms. Hence, the overall moments is the summation of the moments of each Gaussian mechanism by simply applying Lemma~\ref{lem:ori_compose}. Formally, we provide proof of Lemma~\ref{lem:compose} as follows.

From Lemma~\ref{lem:ori_compose}, in each sequence $j$, we have
\begin{equation}\label{eq:lem_1}
    \alpha_{\mathcal{M}^j}\leq \sum_{i=1}^T\alpha_{\mathcal{M}^j_i}(\lambda)
\end{equation}
For $j\in[1,k-1]$, i.e., guest parties, we define a mechanism ${\mathcal{\tilde{M}}_i^j}:\prod_{t=1}^{j-1}\mathcal{R}_t\times\mathcal{D}\rightarrow\mathcal{R}_j$ as follows. We use $o^{1:j}$ to denote $o^1,...,o^j$ for simplicity.
$$\mathcal{\tilde{M}}_i^j(D,o^{1:j})=\mathcal{M}_i^j(D) $$
This mechanism ${\mathcal{\tilde{M}}_i^j}$ satisfies the condition in Lemma~\ref{lem:ori_compose} while it has the same output as $\mathcal{M}_i^j$. For $j=k$, i.e., host party, $\mathcal{M}_i^k$ has already satisfied the condition in Lemma~\ref{lem:ori_compose}. Thus, by applying Lemma~\ref{lem:ori_compose} again on all the parties, we have
\begin{equation}\label{eq:lem_2}
  \alpha_{\mathcal{M}}(\lambda)\leq \sum_{j=1}^k\alpha_{\mathcal{M}^j}(\lambda)  
\end{equation}
Combine (\ref{eq:lem_1}) and (\ref{eq:lem_2}),
$$ \alpha_\mathcal{M}(\lambda)\leq \sum_{j=1}^k\sum_{i=1}^T\alpha_{\mathcal{M}_i^j}(\lambda) $$
\end{proof}
With Lemma~\ref{lem:compose}, the rest of the proof of Theorem~\ref{thm:main} largely follows the proof of moments accountant \cite{abadi2016deep} by replacing $T$ with $kT$. 

\begin{proof}
From Lemma~\ref{lem:gaussian}, since $\lambda\geq 1,q<1$, we have
\[
    \begin{aligned}
        \alpha_{\mathcal{M}_i^j}(\lambda) & \leq \frac{q^2\lambda(\lambda+1)}{(1-q)\sigma^2} \\
        & \leq \frac{q^2\lambda^2}{\sigma^2}\cdot\frac{\lambda+1}{\lambda(1-q)} \\
        & \leq \frac{q^2\lambda^2}{\sigma^2}
    \end{aligned}
\]
From Lemma~\ref{lem:compose}, the overall maximum $\lambda$-th moments can be calculated by a summation.
\begin{equation}\label{eq:moments_all}
    \alpha_{\mathcal{M}}(\lambda)\leq \sum_{j=1}^k\sum_{i=1}^{T_j}\alpha_{\mathcal{M}_i^j}(\lambda) \leq \frac{q^2\lambda^2\sum_{j=1}^k{T_j}}{\sigma^2}
\end{equation}
Given (\ref{eq:moments_all}), considering Lemma~\ref{lem:tail}, in order to ensure the $(\varepsilon,\delta)$-differential privacy of FedOnce-L1, it suffices to prove 
\begin{equation}\label{eq:final_1}
\frac{q^2\lambda^2\sum_{j=1}^k{T_j}}{\sigma^2}\leq \frac{\lambda\varepsilon}{2} 
\end{equation}
\begin{equation}\label{eq:final_2}
    \exp\left(-\frac{\lambda\varepsilon}{2}\right)\leq \delta
\end{equation}
From the condition of Lemma~\ref{lem:gaussian}, we also require
\begin{equation}\label{eq:final_3}
    \lambda\leq \sigma^2\ln\frac{1}{q\sigma}
\end{equation}
When $\varepsilon<c_1q^2\sum_{j=1}^k T_j$, it can be easily verified that there exist constants $c_1$ and $c_2$ that satisfy conditions (\ref{eq:final_1}), (\ref{eq:final_2}), (\ref{eq:final_3}) if we set 
$$ \sigma\geq c_2\frac{q\sqrt{log(1/\delta)\sum_{j=1}^k{T_j}}}{\varepsilon} $$

\end{proof}

With this theorem, the inter-party privacy loss accumulates much slower than simple division. In our experiments, moments division reduces the overall privacy loss by 88.9\% when $k=100$ compared to simple division.

\subsection{Privacy of Representations} 
In this study, assuming the host party is trusted among all the parties, we enable collaborative training by representation sharing to comply with GDPR \cite{voigt2017eu}. Nonetheless, in some applications, this setting may not be true and the representations are potentially threatened by privacy attacks. Such an attack could be similarly designed as the reconstruction attack~\cite{geiping2020inverting} in horizontal federated learning. This privacy risk indicates merely prohibiting unnecessary data sharing is not sufficient to guarantee privacy. More strict and clear regulations above GDPR are desired to measure the permitted information leakage.

To address this privacy issue, some techniques adopted to protect gradients such as local differential privacy (LDP) \cite{kairouz2014extremal}, homomorphic encryption (HE) \cite{cheng2019secureboost}, and secure multi-party computation (MPC) \cite{wu2020privacy,fu2021vf2boost} can potentially assist. Among these techniques, HE and MPC incur large computational cost, which is a heavy burden for devices with limited computational resources like smart Wi-Fi routers. LDP inserts noise into transferred gradients to guarantee the privacy of each party. Nevertheless, representations in FedOnce, in the format of matrices, cannot be similarly analyzed like gradients. Another approach is to regard the representations as released data and adopt privacy-preserving data release \cite{kotsogiannis2019privatesql} to add noise to the representations. As elaborated in Section~\ref{subsec:noise_repr}, this method is a possible solution, whereas the noise on representations should be strictly bounded to a small scale; otherwise, FedOnce could suffer considerable performance loss. In conclusion, more advanced methods are desired to protect the privacy representations in this new scenario where the host party cannot be trusted.

\section{Experiment}\label{sec:experiment}

\subsection{Experiment Setup}\label{subsec:exp_setup}
\noindent\textbf{Datasets.} We evaluate FedOnce on eight public datasets and two real-world federated datasets, whose details are summarized in Table~\ref{tab:dataset}. \textbf{Public datasets}: Public datasets are used to simulate the scenario where each party has the same number of features. \textit{gisette}, \textit{phishing}, and \textit{covtype} are binary classification datasets obtained from LIBSVM\footnote{\url{https://www.csie.ntu.edu.tw/~cjlin/libsvmtools/datasets}}. \textit{UJIIndoorLoc} and \textit{Superconduct} are regression datasets obtained from UCI\footnote{\url{https://archive.ics.uci.edu/ml/datasets.php}}. \textit{MNIST}, \textit{KMNIST} and \textit{Fashion-MNIST} are multi-class classification datasets obtained from TorchVision\footnote{\url{https://pytorch.org/docs/stable/torchvision/datasets.html}}. \textbf{Real-world federated datasets}:  1) \textit{NUS-WIDE}~\cite{nus-wide-civr09} is a real-world multi-view image datasets, used to simulate a security company training a model based on image data collected from different surveillance cameras. In \textit{NUS-WIDE}, each image is stored in the form of five types of low-level features. \textit{NUS-WIDE} provides 81 labels, among which we pick the most occurring label \texttt{sky} to conduct binary classification. 2) \textit{MovieLens}~\cite{harper2015movielens} is a real-world recommendation dataset, used to simulate a film production company training a recommendation model based on the data from a movie rating website and a movie streaming website. \textit{MovieLens} contains 9992 one-hot identity features and 133 auxiliary features. We set movie ratings as labels and regard the task as regression.

\begin{table}[htpb]
\centering
 \caption{Detailed information for datasets}\label{tab:dataset}
    \begin{subtable}{\linewidth}
    \centering
    \caption{Public datasets}\label{tab:pub_dataset}
    \begin{tabular}{cccc}
        \toprule
        \textbf{Dataset} & \textbf{\#samples} & \textbf{\#features} & \textbf{Task}  \\
        \midrule
        gisette & 6,000 & 5,000 & binary-classification \\
        covtype & 581,012 & 54 & binary-classification \\
        phishing & 11,055 & 68 & binary-classification \\
        UJIIndoorLoc & 21,048 & 529 & regression \\
        Superconduct & 21,263 & 81 & regression \\
        MNIST & 60,000 & 28$\times$28 & multi-classification \\
        KMNIST & 60,000  & 28$\times$28 & multi-classification \\
        Fashion-MNIST & 60,000 & 28$\times$28 & multi-classification \\
        \bottomrule
    \end{tabular}
    
    \end{subtable}
    \par\medskip
    \begin{subtable}{\linewidth}
    \centering
    \caption{Real-world federated datasets}\label{tab:real_dataset}
    \begin{tabular}{cccc}
    \toprule
        \textbf{Dataset} & \textbf{\#samples} & \textbf{\#features} & \textbf{Task}  \\
        \midrule
        NUS-WIDE & 269,648 & 66+144+75+128+225 & bin-class \\
        MovieLens & 1,000,209 & 9992+133 & reg \\
        \bottomrule
    \end{tabular}
    
    \end{subtable}
    
\end{table}

\noindent\textbf{Feature Division.}
\wzm{To evaluate FedOnce under different feature divisions,} we consider two kinds of divisions: 1) \textbf{Equal division} for public datasets: each dataset is divided into $k$ parties equally by features. 2) \textbf{Real-world division} for real-world federated datasets: \wzm{each dataset is naturally distributed on $k$ parties by features according to real-world settings}. Due to the distinct distribution of each feature, we further study the effect on the performance when the number and importance of features are biased in Section~\ref{subsec:exp_bias}.

In \textbf{equal division}, for \textit{gisette}, \textit{phishing}, \textit{covtype}, \textit{UJIIndoorLoc} and \textit{Superconduct} with 1D features, we equally divide the features into $k$ parties. For \textit{MNIST}, \textit{KMNIST} and \textit{Fashion-MNIST} with 2D image features, \wzm{since neighboring features are correlated,} we split the images with size 28$\times$28 both horizontally and vertically into four parts with size 14$\times$14 and assign each part to a party ($k=4$). This split on images is commonly adopted to evaluate the performance of vertical federated learning \cite{liu2020asymmetrical,chen2020vafl}. In \textbf{real-world division}, for \textit{NUS-WIDE}, we assign five views of each image into five parties, respectively. For \textit{MovieLens}, we assign 9923 one-hot identity features to one party and assign 133 auxiliary features to the other party.

\noindent\textbf{Baselines.}
To evaluate the communication efficiency of FedOnce-L0, we adopt five baselines: 1) Solo: each party trains locally with its own data and real labels. 2) Combine: all the data and labels are trained centrally. 3) SplitNN~\cite{vepakomma2018split}: a vertical federated learning algorithm for neural networks. 4) SecureBoost~\cite{cheng2019secureboost}: a vertical federated learning framework for gradient boosting decision trees (GBDT). Notably, \textbf{VF$^2$Boost~\cite{fu2021vf2boost} and Pivot~\cite{wu2020privacy}, producing the same accuracy as SecureBoost with additional techniques on encryption, are not individually compared}. Since FedOnce-L0 provides no protection on the released model during the training process, for a fair comparison, we ignore the cryptographic overhead when calculating the communication size of SecureBoost. 5) Linear-Combine: a linear model is trained on centralized datasets. Specifically, we train logistic regression for classification tasks and train ridge regression for regression tasks on centralized datasets similarly to Combine. Among these baselines, SecureBoost is not evaluated on \textit{MNIST}, \textit{KMNIST} and \textit{Fashion-MNIST} since GBDT is unsuitable for 2D features.

To evaluate FedOnce-L1, since the analyses of existing studies \cite{lou2018uplink,lou2020uplink,yao2019privacy} are based on different models which cannot be directly applied to our algorithm based on neural networks. For a fair comparison, we extend existing analyses to a baseline approach, named \textit{Priv-Baseline}, by only adopting their common model-agnostic inter-party analysis. The privacy analysis of Priv-Baseline and FedOnce-L1 are both based on moments accountant~\cite{abadi2016deep}. 

\noindent\textbf{Hardware.}
We conduct the experiments on a machine with two AMD EPYC 7543 32-Core CPUs, four NVIDIA A100 GPUs, and 504GB memory. We implement FedOnce in Python 3.8 and adopt the \textit{pytorch-dp}\footnote{\url{https://github.com/facebookresearch/pytorch-dp}} (now \textit{opacus}) library in FedOnce-L1.

\noindent\textbf{Models.}
The models used for FedOnce-L0 and FedOnce-L1 on each dataset are summarized in Table~\ref{tab:apdx_model_0} and Table~\ref{tab:apdx_model_1}, respectively. The sizes of hidden layers are carefully tuned for each dataset. FC($m\times n$) indicates fully connected layers with two hidden layers which have $m$ and $n$ nodes, respectively. CNN0 and CNN1 indicates two kinds of convolutional neural networks whose structures are shown in Fig.~\ref{fig:apdx_conv1}. NCF($m\times n$) indicates neural collaborative filtering \cite{he2017neural} with two hidden layers which have $m$ and $n$ nodes, respectively.

\begin{table}[htpb]
    \centering
    \caption{Models used in FedOnce-L0}\label{tab:apdx_model_0}
    \begin{tabular}{ccc}
    \toprule
    \textbf{Dataset} & \textbf{Guest Model} & \textbf{Host Model} \\ \midrule
        gisette & FC(100$\times$100$\times$50) & FC(30$\times$30) \\
        covtype & FC(200$\times$100$\times$100) & FC(30$\times$30) \\
        phishing & FC(30$\times$30) & FC(30) \\
        UJIIndoorLoc & FC(50$\times$50$\times$50$\times$30) & FC(30$\times$10) \\
        Superconduct & FC(30$\times$30) & FC(30$\times$10) \\
        MNIST & CNN0$\rightarrow$FC(128) & FC(128) \\
        KMNIST & CNN0$\rightarrow$FC(128) & FC(128) \\
        Fashion-MNIST & CNN0$\rightarrow$FC(128) & FC(128) \\
        NUS-WIDE & FC(60) & FC(50) \\
        MovieLens & NCF(32$\times$16)/NCF(128) & FC(10) \\ \bottomrule
    \end{tabular}
    
\end{table}

\begin{table}[htpb]
    \centering
    \caption{Models used in FedOnce-L1}\label{tab:apdx_model_1}
    \begin{tabular}{ccc}
    \toprule
    \textbf{Dataset} & \textbf{Guest Model} & \textbf{Host Model} \\ \midrule
        gisette & FC(30) & FC(10) \\
        covtype & FC(50) & FC(30) \\
        phishing & FC(30) & FC(10) \\
        UJIIndoorLoc & FC(50$\times$50) & FC(20$\times$20) \\
        Superconduct & FC(50$\times$50) & FC(20$\times$20) \\
        MNIST & CNN1 & FC(10) \\
        \bottomrule
    \end{tabular}
\end{table}

\begin{figure}[htpb]
    \centering
    \begin{subfigure}{.3\linewidth}
      \includegraphics[width=\linewidth]{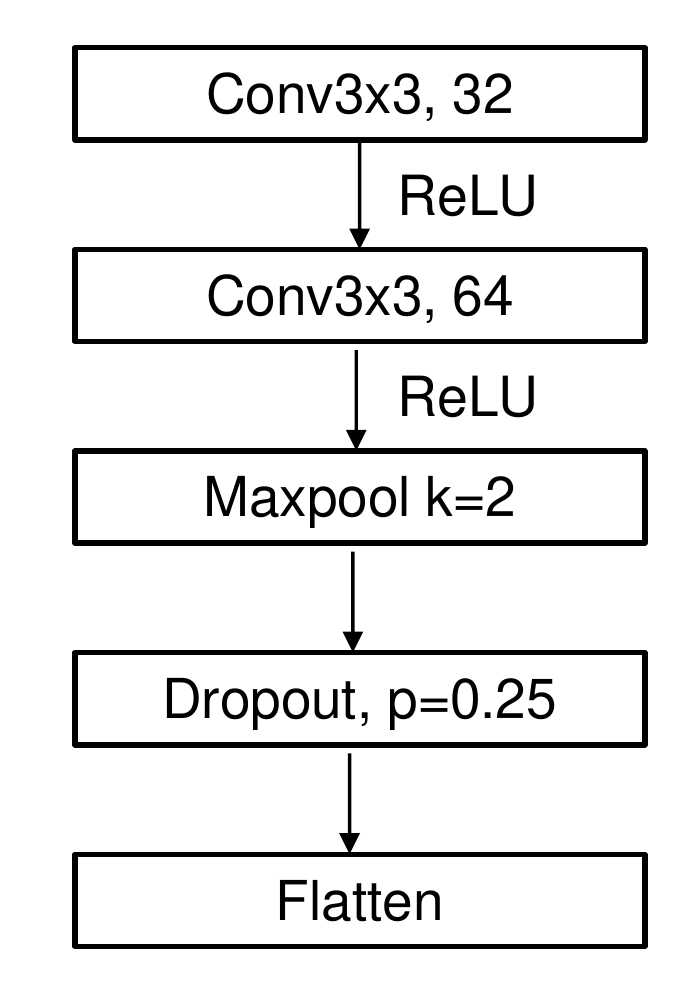}
      \caption{CNN0}
    \end{subfigure}
    \begin{subfigure}{.3\linewidth}
      \includegraphics[width=\linewidth]{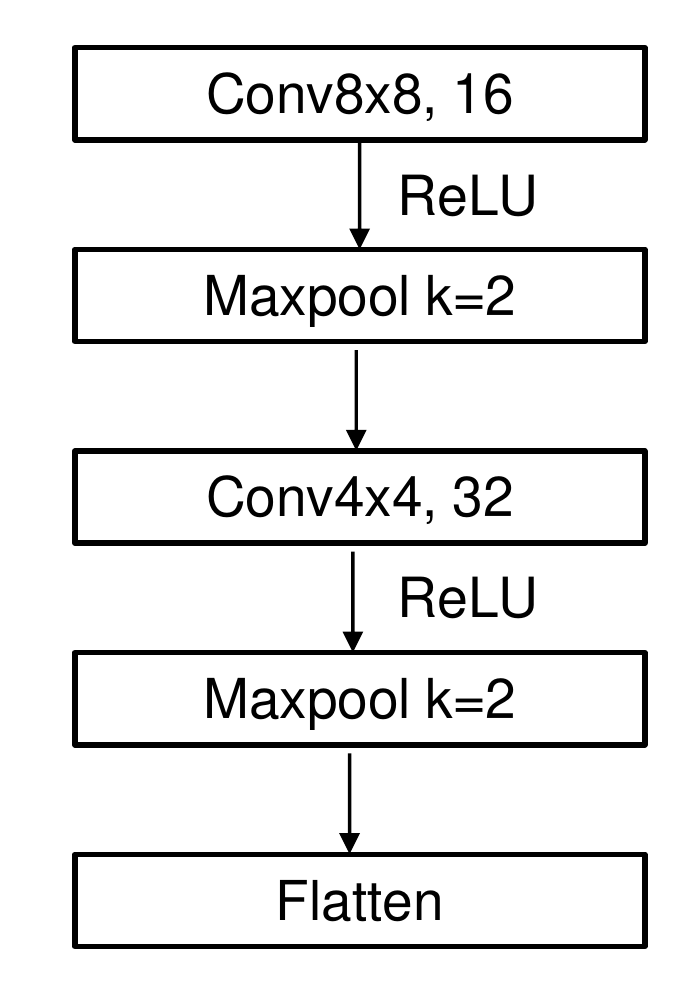}
      \caption{CNN1}
    \end{subfigure}
    \caption{Strucuture of CNN in experiments}
    \label{fig:apdx_conv1}
\end{figure}

\noindent\textbf{Training.}
For \textit{MNIST}, \textit{KMNIST}, \textit{Fashion-MNIST} and \textit{NUS-WIDE}, public test set is used for testing. In other datasets, training and test set are split by 9:1. Five-fold cross-validation is performed for hyperparameter tuning and the test performance is presented. We report accuracy for classification tasks and root mean square error (RMSE) between the real labels and the predicted labels for regression tasks.

In Section~\ref{subsec:exp_perf}, we present the performance of FedOnce-L0 when labels are assigned to each party. In Section~\ref{subsec:exp_comm}, the communication efficiency of FedOnce-L0 is evaluated. In Section~\ref{subsec:scalability}, the scalability of FedOnce is evaluated on a high-dimensional dataset. In Section~\ref{subsec:exp_bias}, we study the performance of FedOnce on biased feature split. In Section~\ref{subsec:apdx_unsup}, we study the effect of different unsupervised learning methods on FedOnce. In Section~\ref{subsec:exp_dp}, we analyze the privacy loss and present the performance of FedOnce-L1 under different privacy budget $\varepsilon$. In Section~\ref{subsec:noise_repr}, we demonstrate the impact of noise adding to representations on the performance.

\begin{figure*}[t!]
    \centering
    \begin{subfigure}{0.58\linewidth}\centering
  \includegraphics[width=.32\linewidth]{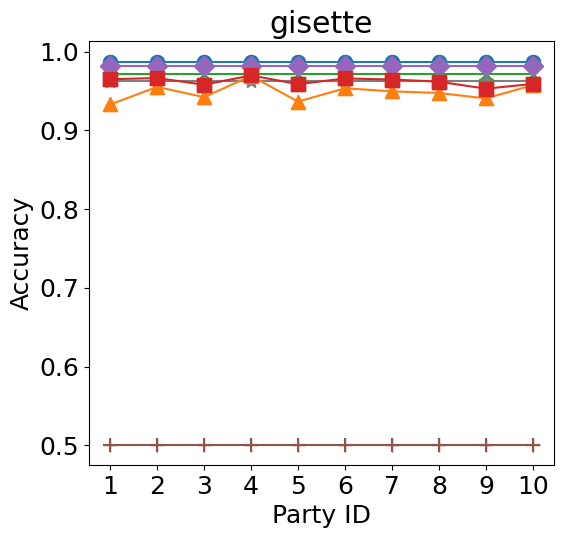}
  \includegraphics[width=.32\linewidth]{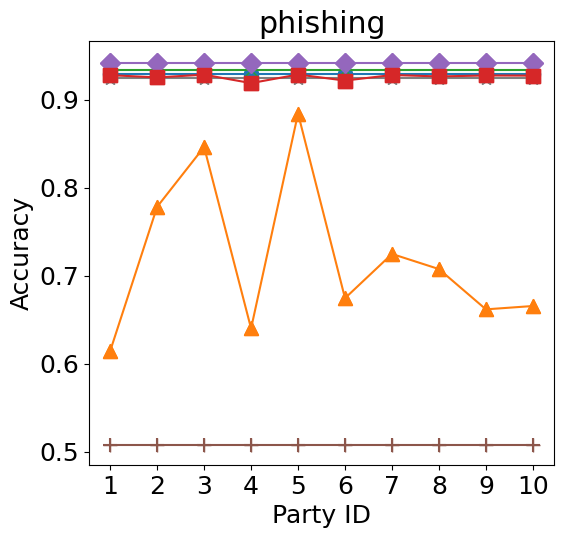}
  \includegraphics[width=.32\linewidth]{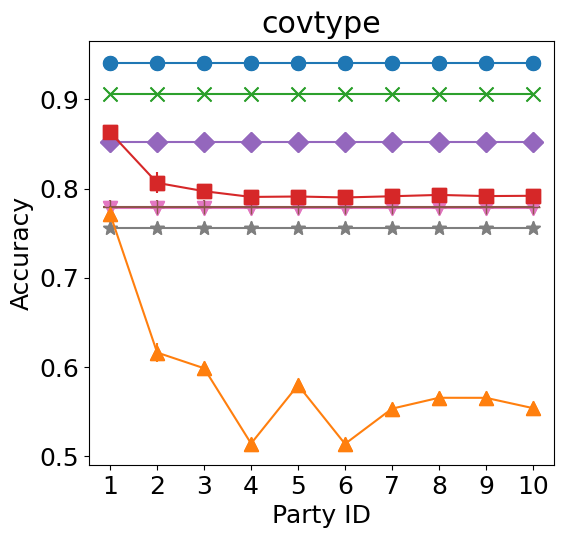}
  \includegraphics[width=\linewidth]{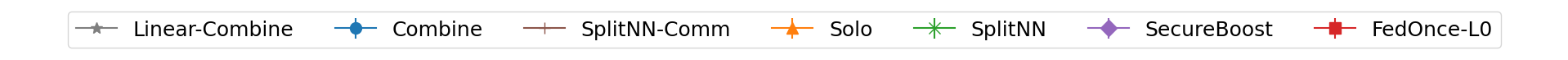}
  \caption{Binary classification datasets ($k=10$)}
  \end{subfigure}
  \medskip
  \begin{subfigure}{.38\linewidth}\centering
  \includegraphics[width=.48\linewidth]{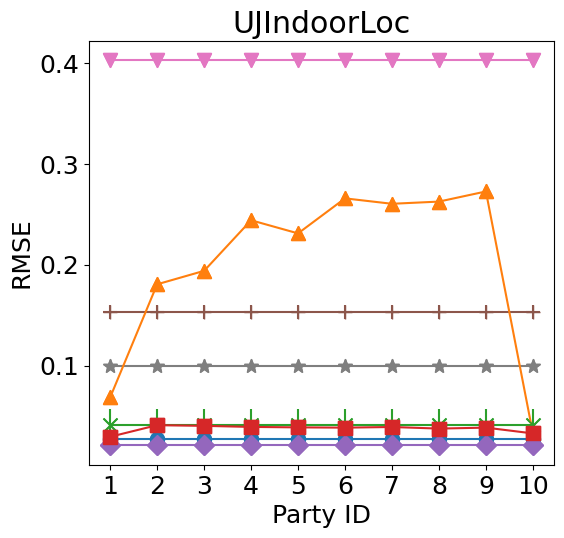}
  \includegraphics[width=.48\linewidth]{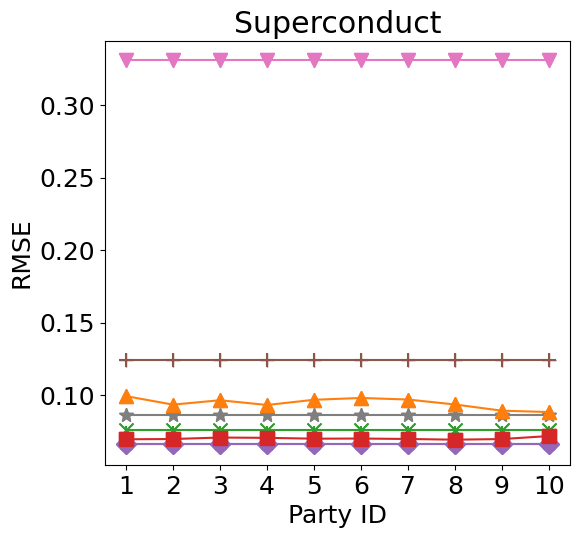}
  \includegraphics[width=\linewidth]{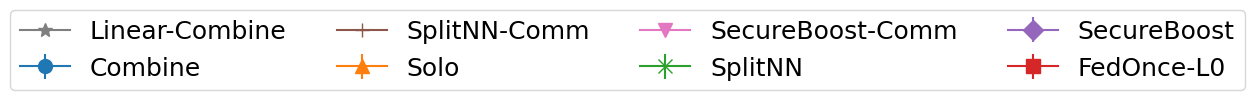}
  \caption{Regression datasets ($k=10$)}
  \end{subfigure}
  \medskip
  \begin{subfigure}{.58\linewidth}\centering
  \includegraphics[width=.32\linewidth]{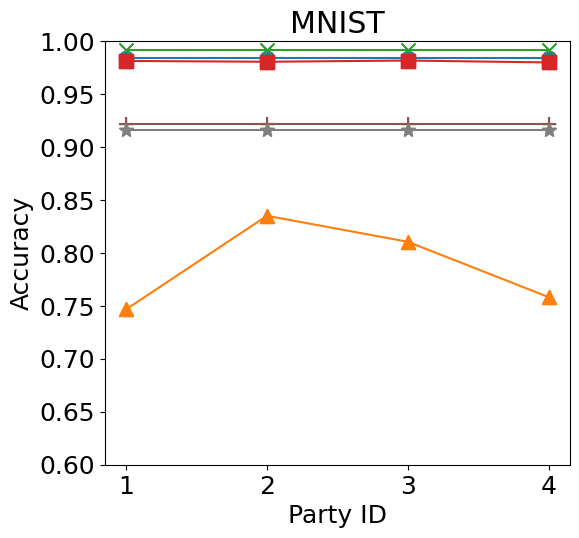}
  \includegraphics[width=.32\linewidth]{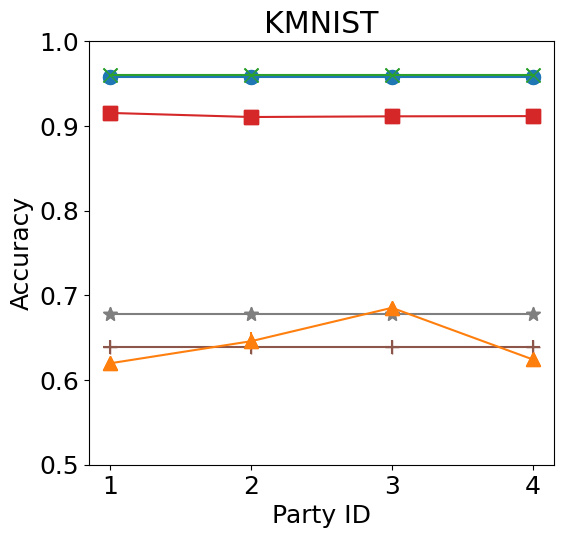}
  \includegraphics[width=.32\linewidth]{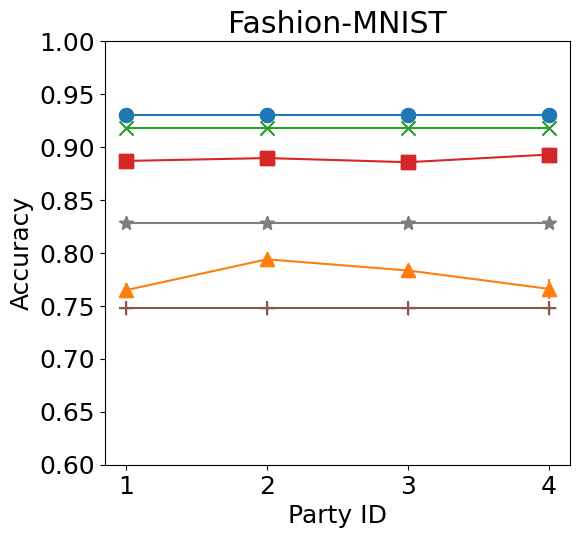}
  \includegraphics[width=.9\linewidth]{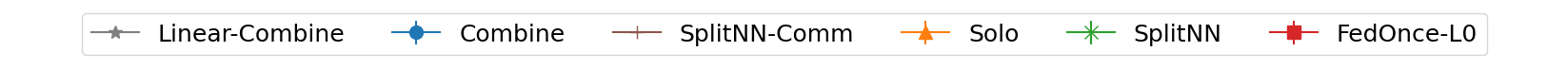}
  \caption{Multi-class classification datasets ($k=4$)}
  \end{subfigure}
 \medskip
  \begin{subfigure}{.38\linewidth}\centering
    \includegraphics[width=.48\linewidth]{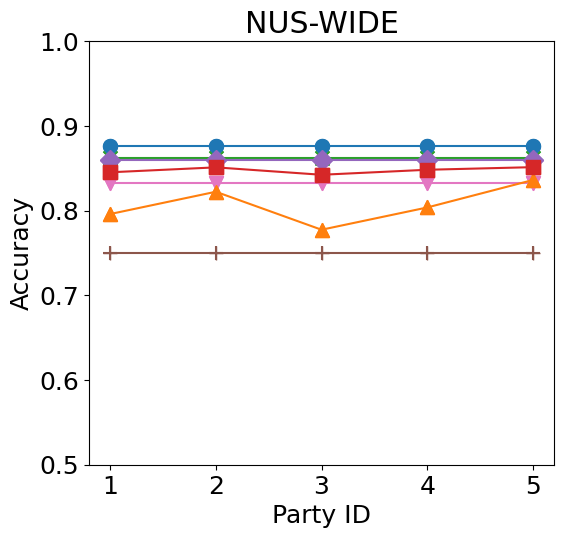}
    \includegraphics[width=.48\linewidth]{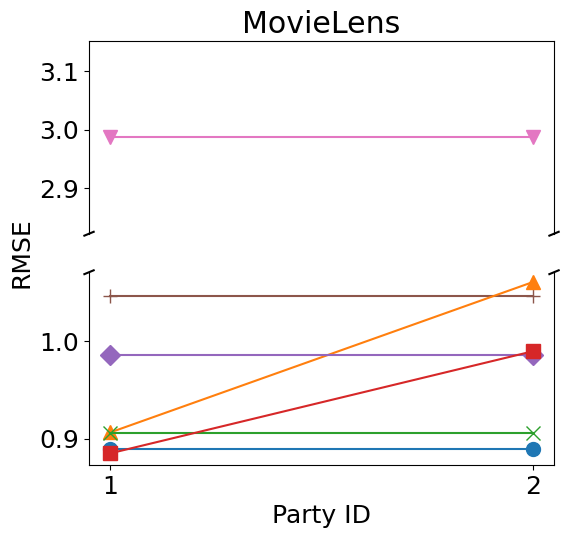}
    \includegraphics[width=\linewidth]{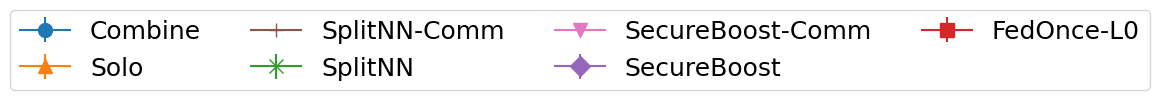}
    \caption{Real-world federated datasets \wzm{($k\in\{5,2\}$)}}
    \end{subfigure}
    \caption{Performance of FedOnce-L0 (i.e., with one communication round)}
    \label{fig:exp_perf_comm}
\end{figure*}

\subsection{Performance of Different Host Parties}\label{subsec:exp_perf}
In this subsection, we evaluate the performance of FedOnce-L0 when each party is chosen as the host party. For a fair comparison, we include two additional baselines: a) SplitNN-Comm: SplitNN terminates training when it costs the same amount of communication size as FedOnce-L0. b) SecureBoost-Comm: SecureBoost terminates training when it costs the same amount of communication size as FedOnce-L0. For FedOnce-L0, Solo, and SecureBoost(-Comm), $j$ indicates the party with labels; for SplitNN(-Comm) and Combine, $j$ is meaningless because each party is treated equally in these settings. The results are displayed in Fig.~\ref{fig:exp_perf_comm}. Since failing to finish the first iteration under the same communication size as FedOnce, SecureBoost-Comm is missing in \textit{gisette} and \textit{phishing}.

Three observations can be made from Fig.~\ref{fig:exp_perf_comm}. First, FedOnce-L0 significantly outperforms Solo, SplitNN-Comm, and SecureBoost-Comm on each party, indicating a huge performance-boosting against SplitNN and SecureBoost under the same communication size. Second, for each party $j$, FedOnce-L0 achieves close performance compared to SplitNN, SecureBoost, and even Combine. For example, FedOnce-L0 only incurs 0.3\% accuracy loss compared to Combine on \textit{MNIST}, and incurs 0.09\% accuracy loss compared to Combine on \textit{phishing}. Third, except on simple datasets like \textit{gisette} and \textit{phishing}, centralized linear models are significantly outperformed by FedOnce and most other baselines. This indicates that vertical federated learning methods with these linear models, which produce even poorer performance, are not competitive in most cases.

\begin{figure*}[htpb]
    \centering
      \begin{subfigure}{0.58\linewidth}\centering
  \includegraphics[width=.32\linewidth]{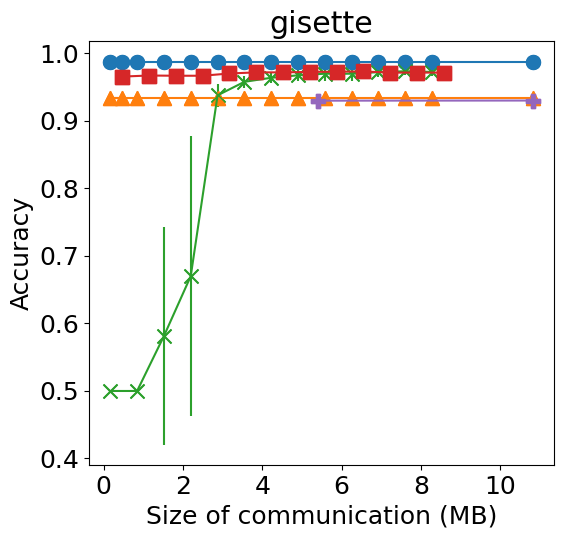}
  \includegraphics[width=.32\linewidth]{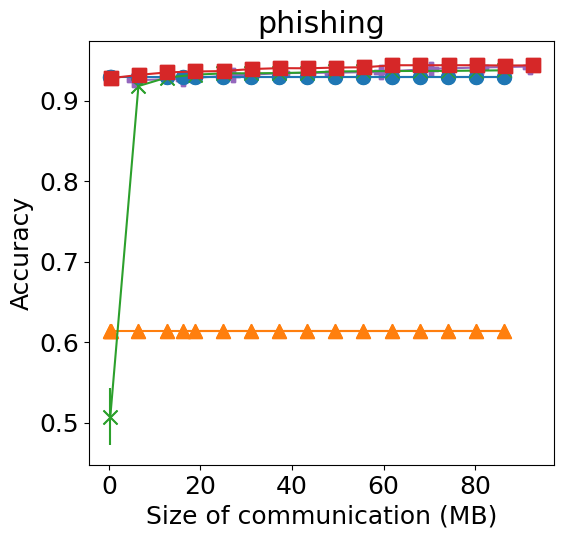}
  \includegraphics[width=.32\linewidth]{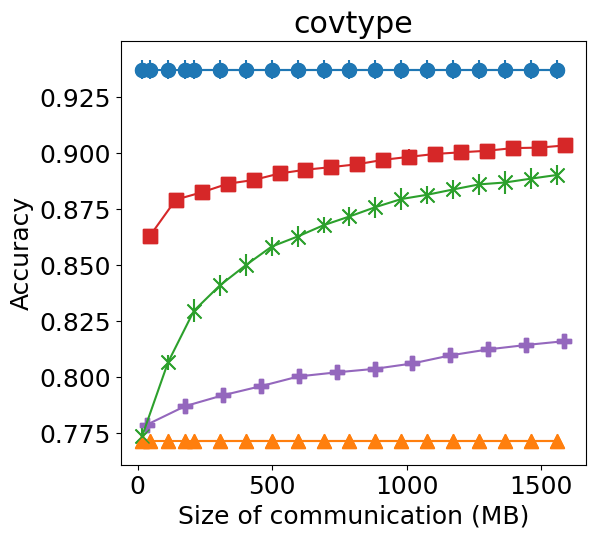}
  \includegraphics[width=.9\linewidth]{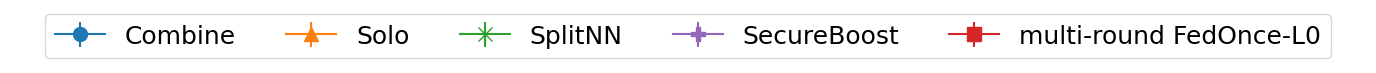}
  \caption{Binary Classification Datasets ($k=10$)}
  \end{subfigure}
  \medskip
  \begin{subfigure}{.38\linewidth}\centering
  \includegraphics[width=.48\linewidth]{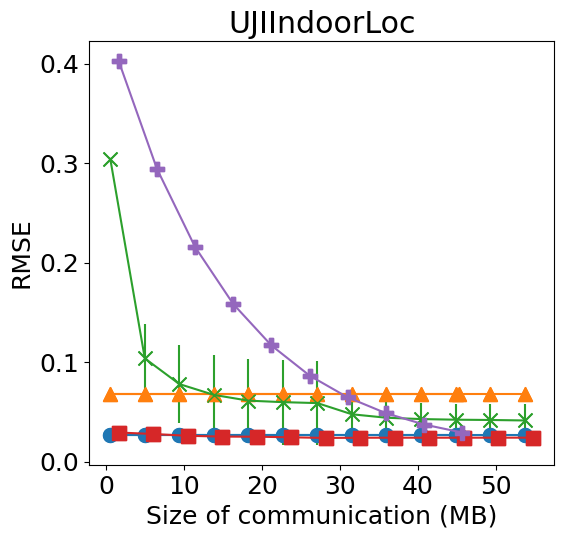}
  \includegraphics[width=.48\linewidth]{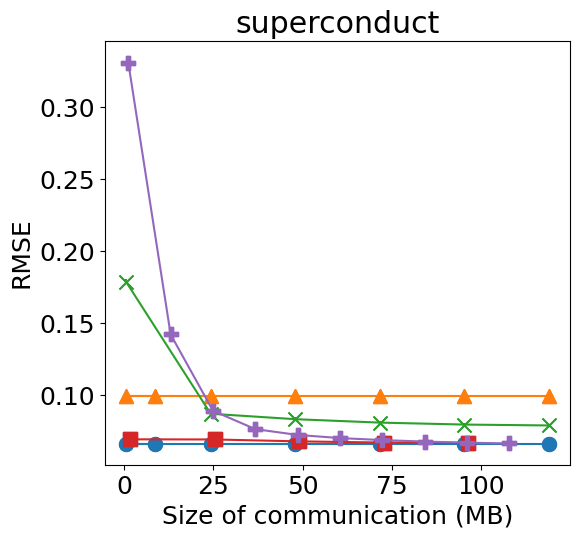}
  \includegraphics[width=.9\linewidth]{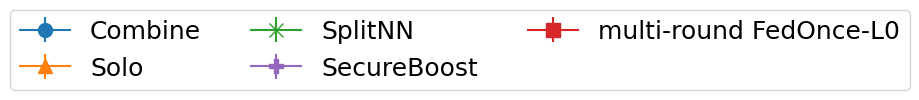}
  \caption{Regression Datasets ($k=10$)}
  \end{subfigure}
  \medskip
  \begin{subfigure}{.58\linewidth}\centering
  \includegraphics[width=.32\linewidth]{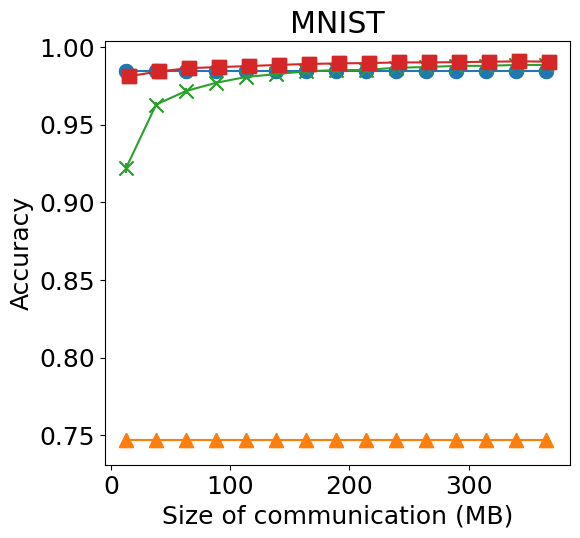}
  \includegraphics[width=.32\linewidth]{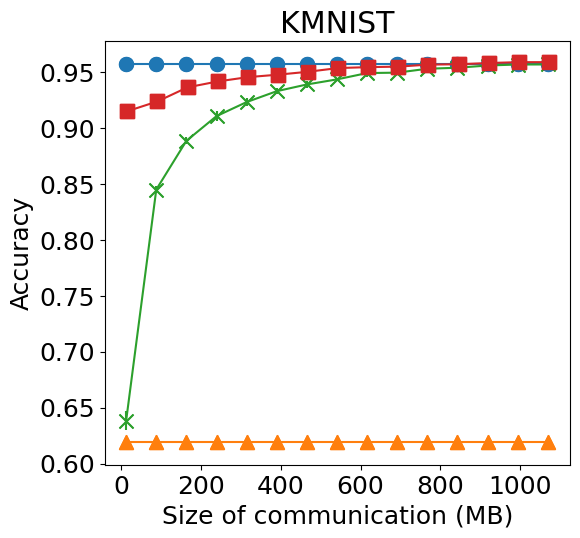}
  \includegraphics[width=.32\linewidth]{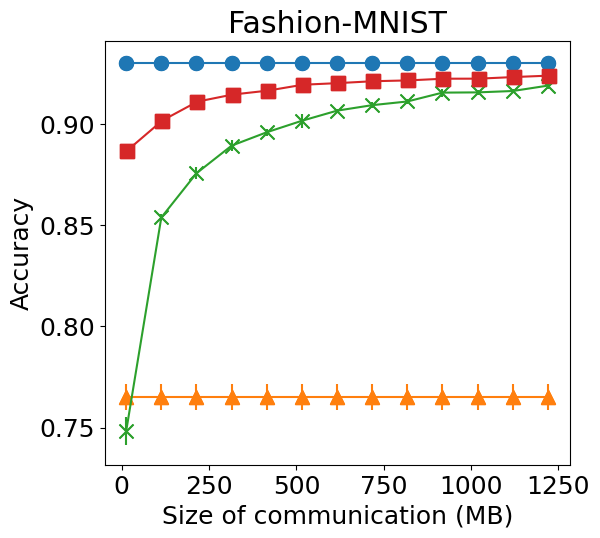}
  \includegraphics[width=.7\linewidth]{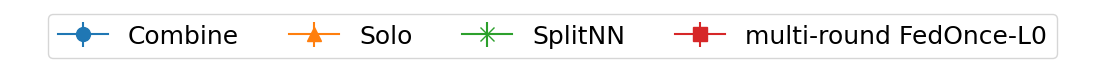}
  \caption{Multi-class Classification Datasets ($k=4$)}
  \end{subfigure}
 \medskip
  \begin{subfigure}{.38\linewidth}\centering
    \includegraphics[width=.48\linewidth]{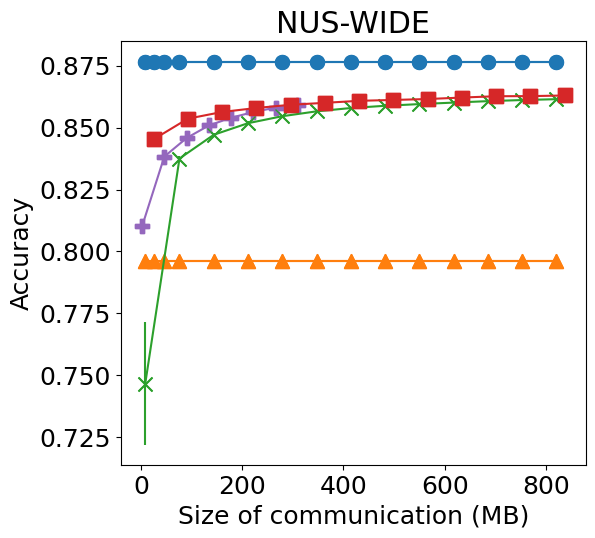}
    \includegraphics[width=.48\linewidth]{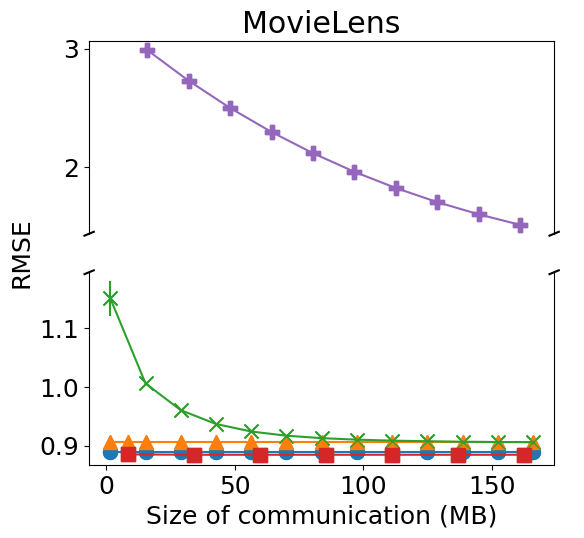}
    \includegraphics[width=.9\linewidth]{figure/post_comm_legend_5_double_col.png}
    \caption{Real-world federated datasets \wzm{($k\in\{5,2\}$)}}
    \end{subfigure}
    \vspace{-5pt}
    \caption{Performance of multi-round FedOnce-L0 and FedOnce-L0 (i.e., the first data point of multi-round FedOnce-L0)}
    \label{fig:apdx_post_comm}
\end{figure*}

\noindent\textbf{Hyperparameters} For each dataset, we summarize the hyperparameters of FedOnce-L0 in Table~\ref{tab:hyper_l0} and hyperparameters of FedOnce-L1 in Appendix~A. In the table, $\eta$ refers to the learning rate, $\lambda$ refers to weight decay, $b$ refers to batch size, $T$ refers to the number of epochs, $d$ refers to the dimension of representations. $f$ indicates the frequency of permutation matrix $P$ to be updated. For example, if the update frequency is 3, $P$ will be updated every three epochs.

Particularly, we study the effect of an important hyperparameter, the dimension of random targets $d$, in Fig.~\ref{fig:exp_perf_hyper} in which the mean and standard variance of the performance when each party is chosen as host party is reported. Intuitively, $d$ reflects the amount of information that can be transferred between the host party and guest parties. A small $d$ causes a less representative $R_j$, leading to less contribution of the data on guest parties to the aggregation model. A large $d$ causes a high-dimensional searching space of representations, among which it might be infeasible to locate the optimal $R_j$. Therefore, a moderate $d$ leads to the best performance of FedOnce-L0. From Fig.~\ref{fig:exp_perf_hyper}, we observe that FedOnce-L0 achieves the best performance when $d\approx16$ on \textit{MNIST}. The choices of $d$ on other datasets are similarly made.

\begin{table}[htpb]
    \centering
    \caption{Training time of FedOnce-L0 and SplitNN}
    \label{tab:train_time}
    \begin{tabular}{ccc}
    \toprule
        \multirow{2}{*}{\textbf{Dataset}} & \multicolumn{2}{c}{\textbf{Training Time(min)}} \\
        \cmidrule{2-3}
         & \textbf{SplitNN}  & \textbf{FedOnce-L0} \\
         \midrule
        gisette & 0.62 & 0.60 \\
        covtype & 171.12 & 52.33  \\ 
        phishing & 4.07 & 1.63  \\
        UJIIndoorLoc & 23.15 & 8.96 \\
        Superconduct & 16.55 & 7.54 \\
        MNIST & 53.58 & 24.17 \\
        KMNIST & 56.47 & 24.13 \\
        Fashion-MNIST & 54.58 & 23.00 \\
        NUS-WIDE & 8.86 & 5.17 \\
        MovieLens & 39.88 & 34.91 \\
         \bottomrule
    \end{tabular}
    
\end{table}

\begin{figure}[htpb]
    \centering
    \includegraphics[width=.75\linewidth]{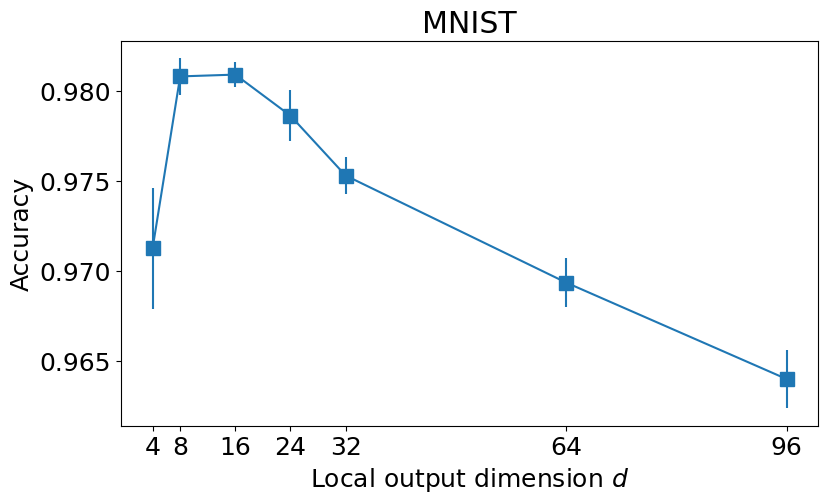}
    \caption{Performance of FedOnce-L0 with different $d$}
    \label{fig:exp_perf_hyper}
    \vspace{-5pt}
\end{figure}

\begin{table*}[htpb]
    \begin{minipage}{\linewidth}
    \centering
    \caption{Hyperparameters of FedOnce-L0 on each dataset}
    \label{tab:hyper_l0}
    \begin{tabular}{ccccccccccccc}
    \toprule
        \multirow{2}{*}{\textbf{Dataset}} & \multirow{2}{*}{\makecell{\textbf{Host}\\\textbf{Party}\textsuperscript{\rm 1}}} & \multicolumn{4}{c}{\textbf{Guest Model}} & \multicolumn{4}{c}{\textbf{Host Model}} & \multirow{2}{*}{$d$} & \multirow{2}{*}{$f$} & \multirow{2}{*}{\textbf{Optimizer}}\\ \cmidrule(lr){3-6}\cmidrule(lr){7-10}
         & & $\eta$ & $\lambda$ & $b$ & $T$ & $\eta$ &  $\lambda$ & $b$ & $T$ \\
         \midrule
         gisette & * & 1e-4 & 1e-5 & 100 & 100 & 1e-4 & 1e-4 & 100 & 100 & 3 & 1 & Adam\\
         covtype & * & 1e-3 & 1e-5 & 100 & 30 & 1e-3 & 1e-3 & 100 & 100 & 3 & 1 & Adam \\
         phishing & * & 1e-4 & 1e-5 & 100 & 100 & 1e-4 & 1e-4 & 100 & 300 & 3 & 1 & Adam \\
         UJIIndoorLoc & * & 3e-4 & 1e-4 & 100 & 500 & 1e-4 & 1e-4 & 100 & 600 & 3 & 1 & Adam \\
         Superconduct & * & 3e-4 & 1e-4 & 100 & 500 & 3e-4 & 1e-4 & 128 & 600 & 3 & 1 & Adam \\
         MNIST & * & 1e-4 & 1e-5 & 128 & 100 & 1e-3 & 1e-5 & 128 & 200 & 16 & 3 & Adam \\
         KMNIST & * & 1e-4 & 1e-5 & 128 & 100 & 1e-3 & 1e-5 & 128 & 200 & 16 & 3 & Adam \\
         Fashion-MNIST & * & 1e-4 & 1e-5 & 128 & 100 & 1e-3 & 1e-5 & 128 & 200 & 16 & 3 & Adam \\
         NUS-WIDE & * & 3e-4 & 1e-5 & 128 & 50 & 7e-5 & 1e-5 & 128 & 50 & 8 & 1 & Adam \\
         \multirow{2}{*}{MovieLens} & 0 & 1e-4 & 1e-5 & 64 & 30 & 1e-4 & 2e-4 & 64 & 40 & 3 & 1 & Adam \\
         & 1 & 1e-4 & 1e-5 & 64 & 30 & 5e-4 & 1e-5 & 64 & 100 & 48 & 1 & Adam \\
         \bottomrule
    \end{tabular}
    \end{minipage}
    
    \textsuperscript{\rm 1}The party that is set as host party. $j$: the hyperparameters when party $j$ is set as host party. * means the hyperparameters remain the same when each party is set as host party.
    
\end{table*}

\noindent\textbf{Training Time.}
For each dataset, we record the training time of SplitNN and FedOnce-L0 in Table~\ref{tab:train_time}. As observed from the table, FedOnce-L0 generally has a lower training time than SplitNN; this is because the local models in FedOnce can be trained in parallel, while all the local models are connected with the aggregation model and trained in each iteration in SplitNN. We also highlight that the communication time is measured in a shared-memory environment of a single machine. In reality, the communication time can also be a huge burden for SplitNN which requires much more communication size as demonstrated in Fig.~\ref{fig:exp_perf_comm}.

\subsection{Communication Efficiency}\label{subsec:exp_comm}
In this subsection, we evaluate the communication efficiency of FedOnce-L0. Without loss of generality, we set party $\mathcal{P}_1$ as the host party and report the mean and standard variance of performance across five folds. After completing one-shot communication, we continue training the models obtained by FedOnce-L0 and track the mean and standard variance of performance across five folds in each iteration. The results are presented in Fig.~\ref{fig:apdx_post_comm}.

Three observations can be made from the results. First, FedOnce-L0 (the first data point of multi-round FedOnce-L0) significantly outperforms SplitNN and SecureBoost at a small communication size. Second, only under considerable communication size, can SplitNN and SecureBoost achieve the same performance as FedOnce-L0. For example, on \textit{KMNIST} dataset, SplitNN requires a 17x communication size to achieve the same performance as FedOnce-L0. Third, multi-round FedOnce-L0 consistently outperforms SplitNN and SecureBoost as communication size increases.

\begin{table*}[t!]
\caption{Performance of FedOnce-L1 with different $\varepsilon$}\label{tab:exp_dp_perf}
\setlength\tabcolsep{2.6pt}
    \begin{subtable}{.49\textwidth}\centering
    \subcaption{Performance on gisette, $\delta=10^{-4}$}
    \begin{tabular}{c c c c c c c}
        \toprule
        \multirow{2}{*}{\textbf{Algorithm}} & \multirow{2}{*}{$\varepsilon$} & \multicolumn{4}{c}{\textbf{Accuracy}} \\\cmidrule{3-7}
        & & \textbf{Party 1} & \textbf{Party 2} & \textbf{Party 3} & \textbf{Party 4} & \textbf{Range} \\
        \midrule
         \multirow{4}{*}{Priv-Baseline} & 2 & 86.20\% & 86.09\% & 75.14\% & 73.00\% & 73.00\%$\sim$86.20\% \\
         & 4 & 89.31\% & 88.10\% & 80.67\% & 80.57\% & 80.57\%$\sim$89.31\% \\
         & 6 & 89.49\% & 88.46\% & 84.01\% & 82.08\% & 82.08\%$\sim$89.49\% \\
         & 8 & 89.71\% & 88.42\% & 83.13\% & 82.55\% & 82.55\%$\sim$89.71\% \\ \midrule
         \multirow{5}{*}{FedOnce-L1} & 2 & 88.81\% & 88.05\% & 78.88\% & 79.51\% & 78.88\%$\sim$88.81\%  \\
         & 4 & 89.73\% & 88.54\% & 81.77\% & 82.88\% & 81.77\%$\sim$89.73\% \\
         & 6 & 89.25\% & 88.67\% & 83.53\% & 82.77\% &  82.77\%$\sim$89.25\% \\
         & 8 & 93.37\% & 93.01\% & 93.14\% & 93.64\% &  93.01\%$\sim$93.64\% \\ 
         & $\infty$ & 96.91\% & 97.11\% & 97.61\% & 96.42\% & 96.42\%$\sim$97.61\% \\
         \bottomrule
    \end{tabular}
    \end{subtable}
    \hspace{\fill}\medskip
    \begin{subtable}{.49\textwidth}\centering
    \subcaption{Performance on covtype, $\delta=10^{-5}$}
    \begin{tabular}{c c c c c c c}
        \toprule
        \multirow{2}{*}{\textbf{Algorithm}} & \multirow{2}{*}{$\varepsilon$} & \multicolumn{4}{c}{\textbf{Accuracy}} \\\cmidrule{3-7}
        & & \textbf{Party 1} & \textbf{Party 2} & \textbf{Party 3} & \textbf{Party 4} & \textbf{Range}  \\
        \midrule
         \multirow{4}{*}{Priv-Baseline} & 2 & 74.60\% & 65.88\% & 60.19\% & 66.53\% & 60.19\%$\sim$74.60\% \\
         & 4 & 75.89\% & 69.14\% & 69.88\% & 68.70\% & 68.70\%$\sim$75.89\% \\
         & 6 & 78.27\% & 73.75\% & 73.46\% & 73.87\% & 73.46\%$\sim$78.27\% \\
         & 8 & 79.25\% & 73.96\% & 73.40\% & 74.04\% & 73.40\%$\sim$79.25\% \\ \midrule
         \multirow{5}{*}{FedOnce-L1} & 2 & 76.88\% & 70.74\% & 70.11\% & 71.56\% & 70.11\%$\sim$76.88\% \\
         & 4 & 77.80\% & 73.33\% & 73.71\% & 72.15\% & 72.15\%$\sim$77.80\% \\
         & 6 & 79.23\% & 74.18\% & 73.66\% & 73.98\% & 73.98\%$\sim$79.23\% \\
         & 8 & 80.49\% & 75.04\% & 74.74\% & 74.80\% & 74.74\%$\sim$80.49\% \\ 
       & $\infty$ & 82.24\% & 74.48\% & 74.29\% & 74.06\% & 74.06\%$\sim$82.24\% \\
         \bottomrule
    \end{tabular}
    \end{subtable}

    \begin{subtable}[t]{.49\textwidth}\centering
    \subcaption{Performance on phishing, $\delta=10^{-5}$}
    \begin{tabular}{c c c c c c c}
        \toprule
        \multirow{2}{*}{\textbf{Algorithm}} & \multirow{2}{*}{$\varepsilon$} & \multicolumn{4}{c}{\textbf{Accuracy}} \\\cmidrule{3-7}
        & & \textbf{Party 1} & \textbf{Party 2} & \textbf{Party 3} & \textbf{Party 4} & \textbf{Range} \\
        \midrule
         \multirow{4}{*}{Priv-Baseline} & 2 & 87.69\% & 87.22\% & 83.98\% & 83.64\% & 83.64\%$\sim$87.69\% \\
         & 4 & 89.43\% & 87.49\% & 88.79\% & 87.92\% & 87.49\%$\sim$88.79\%  \\
         & 6 & 90.31\% & 89.51\% & 89.55\% & 89.29\% & 89.29\%$\sim$90.31\% \\
         & 8 & 90.89\% & 90.76\% & 89.69\% & 90.51\% & 89.69\%$\sim$90.89\% \\  \midrule
         \multirow{5}{*}{FedOnce-L1} & 2 & 90.10\% & 89.36\% & 87.78\% & 88.95\% & 87.78\%$\sim$90.10\% \\
         & 4 & 91.10\% & 90.24\% & 90.90\% & 91.16\% & 90.24\%$\sim$91.16\%  \\
         & 6 & 91.56\% & 90.65\% & 90.36\% & 90.16\% & 90.36\%$\sim$91.56\% \\
         & 8 & 90.78\% & 91.12\% & 90.63\% & 91.30\% & 90.63\%$\sim$91.30\% \\ 
        & $\infty$ & 91.95\% & 91.28\% & 91.21\% & 90.90\% & 90.90\%$\sim$91.95\% \\
         \bottomrule
    \end{tabular}
    \end{subtable}
    \hspace{\fill}\medskip
    \begin{subtable}[t]{.49\textwidth}\centering
    \subcaption{Performance on UJIIndoorLoc, $\delta=10^{-5}$}
    \begin{tabular}{c c c c c c c}
        \toprule
        \multirow{2}{*}{\textbf{Algorithm}} & \multirow{2}{*}{$\varepsilon$} & \multicolumn{4}{c}{\textbf{RMSE}} \\\cmidrule{3-7}
        & & \textbf{Party 1} & \textbf{Party 2} & \textbf{Party 3} & \textbf{Party 4} & \textbf{Range} \\
        \midrule
         \multirow{4}{*}{Priv-Baseline} & 2 & 0.2326 & 0.2275 & 0.2292 & 0.2501 & 0.2275-0.2501 \\
         & 4 & 0.1487 & 0.1578 & 0.2062 & 0.1943 & 0.1487-0.2062 \\
         & 6 & 0.1053 & 0.1224 & 0.1475 & 0.1914 & 0.1053-0.1914 \\
         & 8 & 0.1051 & 0.1198 & 0.1409 & 0.1345 &  0.1051-0.1345 \\ \midrule
         \multirow{5}{*}{FedOnce-L1} & 2 & 0.1066 & 0.1102 & 0.1087 & 0.1489 & 0.1066-0.1489 \\
         & 4 & 0.0769 & 0.0811 & 0.0843 & 0.1362 & 0.0769-0.1362 \\
         & 6 & 0.0714 & 0.0746 & 0.0664 & 0.1223 & 0.0664-0.1223 \\
         & 8 & 0.0542 & 0.0614 & 0.0678 & 0.0730 & 0.0542-0.0730 \\ 
        & $\infty$ & 0.0355 & 0.0369 & 0.0402 & 0.0359 &  0.0355-0.0402 \\
         \bottomrule
    \end{tabular}
    \vspace{-5pt}
    \end{subtable}

    \begin{subtable}[t]{.49\textwidth}\centering
    \subcaption{Performance on MNIST, $\delta=10^{-5}$}
\begin{tabular}{c c c c c c c}
        \toprule
        \multirow{2}{*}{\textbf{Algorithm}} & \multirow{2}{*}{$\varepsilon$} & \multicolumn{4}{c}{\textbf{Accuracy}} \\\cmidrule{3-7}
        & & \textbf{Party 1} & \textbf{Party 2} & \textbf{Party 3} & \textbf{Party 4}& \textbf{Range} \\
        \midrule
         \multirow{4}{*}{Priv-Baseline} & 2 & 45.29\% & 39.98\% & 50.28\% & 44.49\% & 39.98\%$\sim$50.28\% \\
         & 4 & 70.07\% & 76.70\% & 78.00\% & 75.23\% & 70.07\%$\sim$78.00\% \\
         & 6 & 83.64\% & 82.92\% & 84.00\% & 82.60\% &  82.60\%$\sim$84.00\% \\
         & 8 & 85.40\% & 86.33\% & 86.21\% & 85.28\% &  85.28\%$\sim$86.33\%  \\ \midrule
         \multirow{5}{*}{FedOnce-L1} & 2 & 82.13\% & 82.57\% & 83.30\% & 81.43\% & 81.43\%$\sim$83.30\% \\
         & 4 & 87.01\% & 87.51\% & 87.91\% & 87.02\% & 87.01\%$\sim$87.91\%  \\
         & 6 & 88.05\% & 89.03\% & 88.24\% & 87.91\% & 87.91\%$\sim$89.03\% \\
         & 8 & 90.52\% & 90.63\% & 89.61\% & 89.32\% & 89.32\%$\sim$90.63\% \\ 
         & $\infty$ & 92.05\% & 92.97\% & 93.11\% & 91.74\% & 91.74\%$\sim$93.11\% \\
         \bottomrule
    \end{tabular}
    \end{subtable}
    \hspace{\fill}
    \begin{subtable}[t]{.49\textwidth}\centering
    \subcaption{Performance on Superconduct, $\delta=10^{-5}$}
    \begin{tabular}{c c c c c c c}
        \toprule
        \multirow{2}{*}{\textbf{Algorithm}} & \multirow{2}{*}{$\varepsilon$} & \multicolumn{4}{c}{\textbf{RMSE}} \\\cmidrule{3-7}
        & & \textbf{Party 1} & \textbf{Party 2} & \textbf{Party 3} & \textbf{Party 4} & \textbf{Range}\\
        \midrule
         \multirow{4}{*}{Priv-Baseline} & 2 & 0.1668 & 0.1667 & 0.1783 & 0.1373 & 0.1373-0.1783 \\
         & 4 & 0.1253 & 0.1387 & 0.1348 & 0.1246 & 0.1246-0.1387 \\
         & 6 & 0.1173 & 0.1256 & 0.1299 & 0.1230 & 0.1173-0.1299 \\
         & 8 & 0.1125 & 0.1140 & 0.1308 & 0.1148 & 0.1125-0.1308 \\ \midrule
         \multirow{5}{*}{FedOnce-L1} & 2 & 0.1080 & 0.1043 & 0.1059 & 0.1110 & 0.1043-0.1110\\
         & 4 & 0.1071 & 0.1003 & 0.1009 & 0.1052 & 0.1003-0.1052 \\
         & 6 & 0.0985 & 0.1001 & 0.0984 & 0.0973 & 0.0973-0.1001 \\
         & 8 & 0.1020 & 0.1075 & 0.0985 & 0.0956 &  0.0956-0.1075 \\ 
        & $\infty$ & 0.0832 & 0.0833 & 0.0841 & 0.0801 & 0.0801-0.0841 \\
         \bottomrule
    \end{tabular}
    \end{subtable}
     \vspace{-5pt}
\end{table*}

\subsection{Scalability} \label{subsec:scalability}

In this subsection, we evaluate the scalability of FedOnce on a high-dimensional dataset \textit{gisette} with 5000 features. The performance of FedOnce-L0 and baselines when $k$ ranges from $[10,200]$ is displayed in Fig.~\ref{fig:exp_scale}. The error bars indicate the standard variance of performance when labels are assigned to different parties. SecureBoost-Comm is omitted since SecureBoost cannot finish the first iteration with the same communication size as FedOnce.

\begin{figure}[t!]
    \centering
    \includegraphics[width=.7\linewidth]{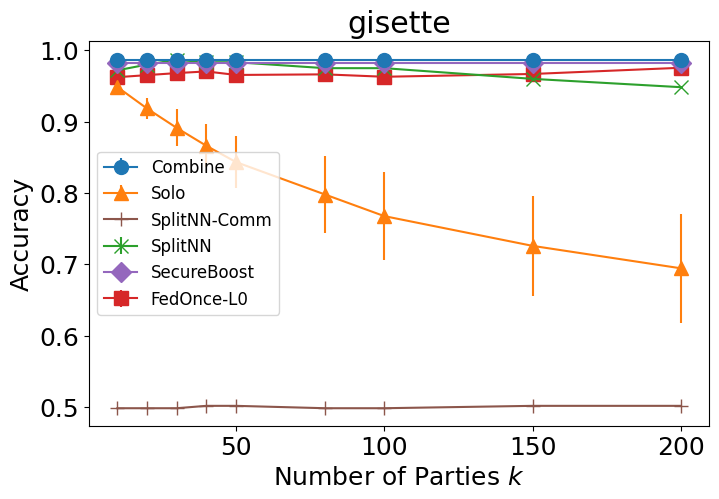}
    \caption{Performance of FedOnce-L0 (i.e., with one communication round) when $k$ scales from 10 to 200}
     \vspace{-5pt}
    \label{fig:exp_scale}
\end{figure}

From Fig.~\ref{fig:exp_scale}, we can make two observations. First, as the number of parties $k$ increases, FedOnce-L0, whose performance remains stable while the performance of Solo degrades rapidly, is scalable. Second, even at a large number of parties, FedOnce consistently outperforms SecureBoost and SplitNN with the same communication size.

\subsection{Performance on Biased Datasets}\label{subsec:exp_bias}

In this subsection, we study how the quality of features affects the performance of FedOnce. For convenience, the experiment is conducted under a two-party setting \wzm{due to the constraint of our designed metric of bias}. First, we leveraged XGBoost\footnote{\url{https://github.com/dmlc/xgboost}} library to calculate the importance of each feature. Then, the top-$p$ important features are marked as \textit{dominant features}. We randomly assign a ratio $\alpha$ of dominant features to the host party and assign the remaining $1-\alpha$ of dominant features to the guest party. The other features are randomly assigned to two parties to ensure both parties hold the same number of features. Finally, we adjust the ratio $\alpha\in[0,1]$ to control the feature quality of the host party. The host party holds the best features when $\alpha=1$ while holding the worst features when $\alpha=0$. We conduct experiments on \textit{phishing} (classification) and \textit{UJIndoorLoc} (regression). Under five-fold cross-validation, the mean performance of the host party as well as the standard variance across five folds are reported in Fig.~\ref{fig:exp_perf_bias}. 
\begin{figure}[htpb]
    \centering
    \includegraphics[width=.48\linewidth]{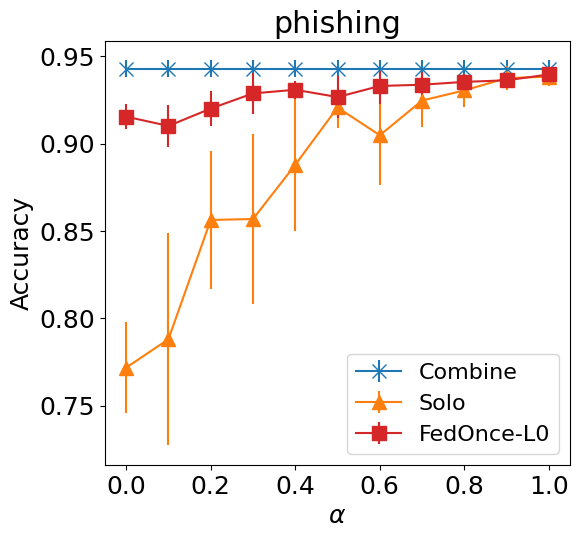}
    \includegraphics[width=.48\linewidth]{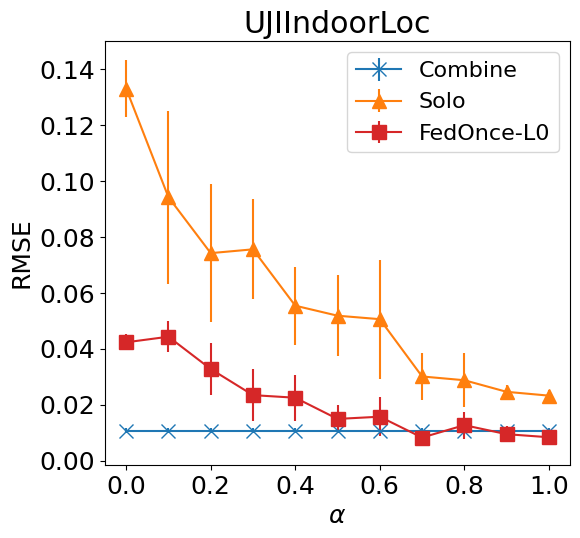}
    \caption{Performance on biased datasets}
    \vspace{-5pt}
    \label{fig:exp_perf_bias}
\end{figure}

Two observations can be made from Fig.~\ref{fig:exp_perf_bias}. First, FedOnce achieves the most significant performance when the quality of features is poor (i.e., $\alpha\rightarrow 0$). Second, when the feature quality of the host party is good enough (i.e., $\alpha\rightarrow 1$), federated learning is no longer necessary since even Solo provides competitive performance compared to Combine.

\subsection{Different Unsupervised Learning Methods}\label{subsec:apdx_unsup}

FedOnce is suitable for different unsupervised learning methods. Among these unsupervised learning methods, we compare the performance of FedOnce with NAT~\cite{bojanowski2017unsupervised} (FedOnce-L0) and Principal Component Analysis~\cite{xu1994theories} (FedOnce-PCA). Specifically, \wzm{fixing both the dimensions of representative features and the number of principal components the same,} we present the results of four typical datasets covering image features and multi-variant features in Table~\ref{tab:apdx_pca}. Without loss of generality, party $\mathcal{P}_1$ is selected as the host party. \wzm{The choice of $K$ is identical to that in Section~\ref{subsec:exp_perf}.}

\begin{table}[htpb]
    \centering
    \caption{Performance of FedOnce on different unsupervised learning methods}
    \begin{tabular}{cccc}
    \toprule
        \multirow{2}{*}{\textbf{Feature Type}} & \multirow{2}{*}{\textbf{Dataset}} & \multicolumn{2}{c}{\textbf{Accuracy}} \\
        \cmidrule{3-4}
       &  & \textbf{FedOnce-L0}  & \textbf{FedOnce-PCA} \\
         \midrule
        \multirow{2}{*}{{Multi-variant}} & gisette & 96.51\% & 94.10\% \\
        & phishing & 92.78\% & 92.41\% \\ \midrule
        \multirow{2}{*}{{Image}} & MNIST & 98.05\% & 25.98\% \\
        & KMNIST & 91.73\% & 25.59\% \\
         \bottomrule
    \end{tabular}
    \label{tab:apdx_pca}
\end{table}

From Table~\ref{tab:apdx_pca}, we observe that FedOnce-PCA can achieve close performance to FedOnce-L0 on multi-variant datasets, but has very poor performance on image datasets. This is because PCA is incapable to extract useful representations from complex features like images. On the contrary, NAT is a suitable unsupervised learning method for FedOnce that can handle different types of features. \wzm{Generally, more ``advanced'' unsupervised learning algorithms tend to produce a better performance on FedOnce. The performance of unsupervised learning methods can be evaluated by commonly used metrics in the literature. For example, NAT is evaluated by the performance of a linear classifier on the learned representations.}

\subsection{Differential Privacy}\label{subsec:exp_dp}
In this subsection, we first compare the privacy loss of FedOnce-L1 and simple division by theoretical analysis, then present how the moments division is superior to simple division on performance.

\noindent\textbf{Privacy Loss Analysis.} Fixed the hyperparameters of each party including the number of samples $n$, batch size $b$, the number of epochs $T$, Gaussian noise multiplier $\sigma$ and overall $\delta$, we increase the number of parties $k$ and calculate the privacy loss of Priv-Baseline and FedOnce-L1. The results are presented in Fig.~\ref{fig:exp_priv_loss}, from which we observe that the privacy loss of FedOnce-L1 increases much slower than Priv-Baseline when $k>1$. Hence, in FedOnce-L1, each party receives a higher privacy budget, thus incurring less performance loss.

\begin{figure}[htbp]
    \centering
    \includegraphics[width=.7\linewidth]{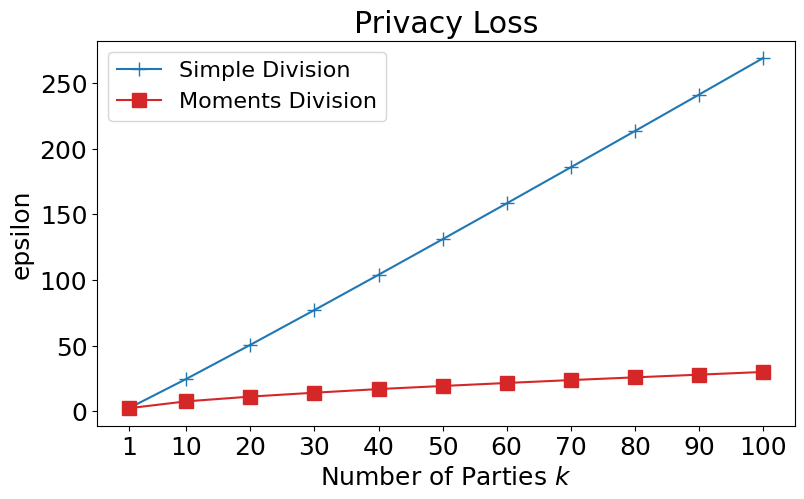}
    \caption{Privacy loss accumulation when $n=60000$, $b=128$, $T=50$, $\sigma=1$, $\delta=10^{-5}$}
    \label{fig:exp_priv_loss}
\end{figure}

\noindent\textbf{Performance.} Despite the much lower privacy loss compared to simple division, FedOnce-L1 still suffers significant performance loss at a large $k$ like other vertical federated learning algorithms. Therefore, we set $k=4$ and report the performance under different overall privacy budgets $\varepsilon$ on six public datasets. \textit{Fashion-MNIST} and \textit{KMNIST}, which fail to produce reasonable accuracy, are not included in this experiment. The results are summarized in Table~\ref{tab:exp_dp_perf}. 

Two observations can be made from Table~\ref{tab:exp_dp_perf}. First, FedOnce-L1 significantly outperforms Priv-Baseline with the same privacy budget $\varepsilon$. Second, compared with FedOnce-L1 with $\varepsilon=\infty$ (no noise is added), FedOnce-L1 has a close performance under a modest $\varepsilon$ in most datasets. Additionally, the performance of FedOnce-L1 with a small $\varepsilon$ (e.g., $\varepsilon=2$) can be significantly affected by the noise, implying more advanced privacy mechanisms are desired.

\subsection{Noise on Representations}\label{subsec:noise_repr}
In this experiment, to preserve the privacy of representations, we investigate the effect of adding noise to representations on the performance of FedOnce-L0. Specifically, independent Gaussian noise of scale $\sigma_r$ is added to all the representations. We report the average performance of all the parties on two real-world federated datasets in Table~\ref{tab:noise_repr}.

\begin{table}[htpb]
    \centering
    \setlength\tabcolsep{4pt}
    \caption{Performance of FedOnce-L0 with Gaussian noise of scale $\sigma_r$ on representations (bold values mean outperforming Solo)}
    \begin{tabular}{c c c c c c c}
    \toprule
    \multirow{3}{*}{\textbf{Dataset}} & \multicolumn{6}{c}{\textbf{Accuracy/RMSE}} \\
        \cmidrule{2-7}
        & \multicolumn{5}{c}{\textbf{FedOnce-L0 w/ different $\sigma_r$}} & \multirow{2}{*}{\textbf{Solo}} \\ \cmidrule{2-6}
        & 0.0 & 0.1 & 0.5 & 1.0 & 1.5 & \\
        \midrule
        NUS-WIDE & \textbf{84.76\%} & \textbf{84.64\%} & \textbf{82.75\%} & \textbf{81.45\%} & 81.05\% & 80.71\% \\
        MovieLens & \textbf{0.9373} & \textbf{0.9578} & \textbf{0.9690} & \textbf{0.9679} & 0.9969 & 0.9835 \\
\bottomrule
    \end{tabular}
    \label{tab:noise_repr}
\end{table}

As can be observed from Table~\ref{tab:noise_repr}, the performance of FedOnce-L0 drops below Solo at a relatively large scale of noise (e.g., $\sigma_r=1.5$) on both datasets. This observation indicates that techniques like privacy-preserving data releasing can potentially be used to protect representations, whereas the added noise must be restricted to a small scale (e.g., $\sigma_r<1.0$).

\section{Related Work}\label{sec:relatedworks}
\begin{table*}[htpb]
    \begin{minipage}{\linewidth}
    \centering
    \caption{Related work in vertical federated learning}\label{tab:related_work}
    \begin{tabular}{c c c c c}
    \toprule
        \textbf{Algorithm} & \textbf{Label Owner\textsuperscript{\rm 1}} & \textbf{Privacy\textsuperscript{\rm 2}} & \textbf{Model\textsuperscript{\rm 3}} & \textbf{Comm. Rounds\textsuperscript{\rm 4}}  \\ \midrule
        FDML~\cite{hu2019fdml} & All & N/A & NN & Multi \\
        SplitNN~\cite{vepakomma2018split} & Single & N/A & NN & Multi \\
        SecureBoost~\cite{cheng2019secureboost} & Single & HE & GBDT & Multi \\
        Pivot\cite{wu2020privacy}\textsuperscript{\rm 5} & Single & MPC/HE & GBDT/LR & Multi \\
        VF$^2$Boost\cite{fu2021vf2boost}\textsuperscript{\rm 5} & Single & HE & GBDT & Multi \\
        \cite{lou2018uplink,lou2020uplink} &  Single & DP-S & FW & Multi \\
        \cite{yao2019privacy} & Single & DP-S & LR/HTL & Multi \\
        \midrule
        \textbf{FedOnce} & \textbf{Single} & \textbf{DP-M} & \textbf{NN} & \textbf{Single} \\
        \bottomrule
    \end{tabular} \\
    \end{minipage}
    
    \textsuperscript{\rm 1} The number of parties that own the labels. Single: Only one party owns the labels, All: all parties own the labels \\
    \textsuperscript{\rm 2} Privacy mechanism used in the algorithm. N/A: Not specified in details, HE: homomorphic encryption, MPC: secure multi-party computation, DP-S: differential privacy with simple division among parties, DP-M: differential privacy with moments division among parties; \\
    \textsuperscript{\rm 3} Model supported by the algorithm. NN: neural networks, GBDT: gradient boosting decision trees, LR: logistic regression, FW: Floyd–Warshall algorithm, HTL: hypothesis transfer learning. \\
    \textsuperscript{\rm 4} Communication rounds required for the algorithm. Multi: multiple rounds are required, Single: single round is required. \\
    \textsuperscript{\rm 5} Pivot~\cite{wu2020privacy} and VF$^2$Boost~\cite{fu2021vf2boost} have the \textit{same} accuracy as SecureBoost with additional techniques on encryption.
    \vspace{-5pt}
\end{table*}

In this section, we review the literature in three aspects and summarize the related work in Table~\ref{tab:related_work}.

\noindent\textbf{Vertical Federated Learning.}
Except for a study \cite{hu2019fdml} that investigates the scenario where labels are shared across all the parties, most existing studies in vertical federated learning focus on another more practical scenario where only one party holds the labels. Some tree-based approaches, including \textit{SecureBoost} \cite{cheng2019secureboost}, Pivot~\cite{wu2020privacy}, and VF$^2$Boost~\cite{fu2021vf2boost}, are proposed to enable vertical federated learning in gradient boosting decision trees \cite{chen2016xgboost}. Besides, federated random forest is also studied in \cite{liu2020federated}. \textit{SplitNN} \cite{vepakomma2018split} is proposed to collaboratively train neural networks by splitting the model among parties and exchanging gradients and outputs in each epoch. \wzm{These approaches that require multi-round communication suffer large communication cost, whereas our proposed FedOnce only requires one-round communication, leading to much smaller communication cost.} Besides trees and neural networks, there are also studies on other machine learning algorithms such as linear regression \cite{zhang2021secure} and logistic regression \cite{hu2019learning,liu2020communication}. However, these algorithms are usually incapable to handle complex tasks.

\noindent\textbf{Privacy in Federated Learning.}
Federated learning confronts two major privacy threats. First, during the training, intermediate results (e.g., gradients) could be vulnerable to backdoor attack \cite{wang2020attack} or reconstruction attack \cite{geiping2020inverting}. This privacy risk can be addressed by multiple methods including homomorphic encryption \cite{cheng2019secureboost}, secure multi-party computation \cite{wu2020privacy,fu2021vf2boost}, and local differential privacy~\cite{kairouz2014extremal}. Second, after the training, the released model could be exposed to the membership inference attack~\cite{Shokri2017MembershipIA}, which can be defended by ensuring differential privacy of the released model (namely \textit{global differential privacy}). This defense is different from local differential privacy since it protects against outside attackers instead of malicious aggregators.

Nevertheless, existing studies on global differential privacy under vertical federated learning suffer significant performance loss. One study \cite{xu2019achieving} achieves differential privacy by \textit{objective perturbation}, which is often intractable in practice compared to \textit{gradient perturbation} according to \cite{wang2018empirical}. The other two studies \cite{lou2018uplink,yao2019privacy} apply differential privacy based on gradient perturbation, but they both focus on inner-party privacy instead of inter-party privacy. Specifically, in both studies, the simple composition is directly applied in the analysis of privacy loss across parties, leading to excessive noise. In this paper, we apply moments accountant to analyze the inter-party privacy loss, thus effectively reducing the overall privacy loss compared to simple composition.

\noindent\textbf{Communication-Efficient Federated Learning.}
Most existing studies in communication-efficient federated learning \cite{bernstein2018signsgd,jin2020stochastic,hamer2020fedboost,guha2019one,zhou2020distilled} focus on horizontal federated learning. These approaches, requiring each party to train independently, cannot be applied to vertical federated learning where only one party holds the labels. Though some approaches \cite{chen2020vafl,feng2020multi,cheng2019secureboost} study the communication efficiency in vertical federated learning, they all require multiple communication rounds and a certain level of synchronization. As observed from Table~\ref{tab:related_work}, vertical federated learning with one-shot communication remains unexplored.
 
\section{Conclusion}\label{sec:conclusion}
In this paper, we propose FedOnce, a novel one-shot algorithm for vertical federated learning. Moreover, we develop a privacy mechanism for FedOnce (FedOnce-L1) under the notion of differential privacy and refine inter-party privacy loss. Our experiments demonstrate that FedOnce-L0 achieves impressive performance compared to state-of-the-art vertical federated learning algorithms with only one-shot communication. Besides, FedOnce-L1 significantly outperforms the baseline in our evaluation.

\ifCLASSOPTIONcompsoc
  \section*{Acknowledgments}
\else
  \section*{Acknowledgment}
\fi

This research is supported by the National Research Foundation, Singapore under its AI Singapore Programme (AISG Award No: AISG2-RP-2020-018). Any opinions, findings and conclusions or recommendations expressed in this material are those of the authors and do not reflect the views of National Research Foundation, Singapore. Qinbin is also in part supported by a Google PhD Fellowship.

\ifCLASSOPTIONcaptionsoff
  \newpage
\fi

\bibliography{reference.bib}
\bibliographystyle{IEEEtran}

\vspace{-20pt}
\begin{IEEEbiography}
[{\includegraphics[width=1in,height=1.25in, clip,keepaspectratio]{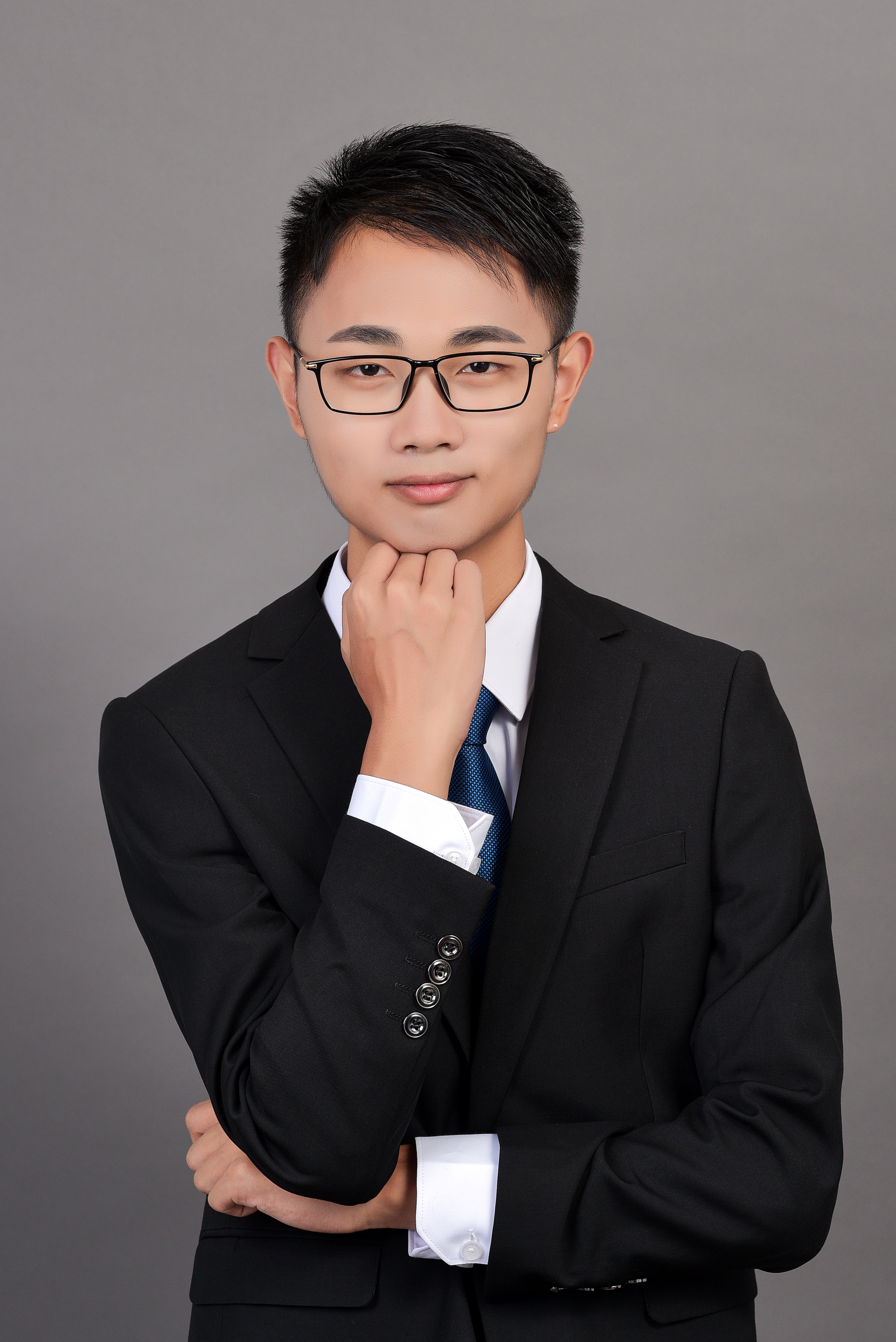}}]{Zhaomin Wu} is a Ph.D. candidate in School of Computing of National University of Singapore. He received the bachelor degree in computer science from Huazhong University of Science and Technology (2015-2019). His current research interests include federated learning and privacy.
\end{IEEEbiography}
\vspace{-20pt}
\vspace{-20pt}
\begin{IEEEbiography}[{\includegraphics[width=1in,height=1.25in, clip,keepaspectratio]{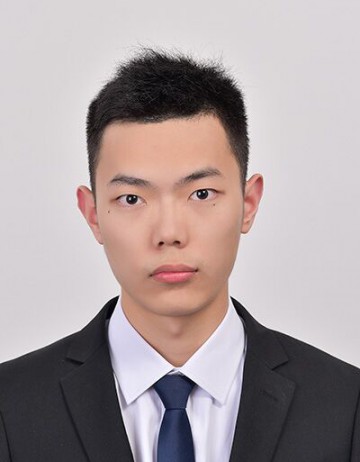}}]{Qinbin Li}
is currently a Ph.D. candidate in School of Computing of National University of Singapore. His current research interests include machine learning, federated learning, and privacy.
\end{IEEEbiography}
\vspace{-20pt}
\vspace{-20pt}
\begin{IEEEbiography}
[{\includegraphics[width=1in,height=1.25in, clip,keepaspectratio]{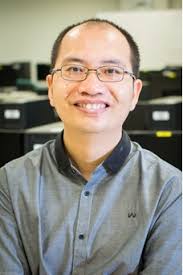}}]{Bingsheng He}
is an Associate Professor in School of Computing of National University of Singapore. He received the bachelor degree in computer science from Shanghai Jiao Tong University (1999-2003), and the PhD degree in computer science in Hong Kong University of Science and Technology (2003-2008). His research interests are high performance computing, distributed and parallel systems, and database systems.
\end{IEEEbiography}
\appendices

\section{Choice of Hyperparameters}\label{apdx:hyper}
For each dataset, we summarize the hyperparameters of FedOnce-L1 in Table~\ref{tab:hyper_l1}. In the table, $\eta$ refers to the learning rate, $\lambda$ refers to weight decay, $b$ refers to batch size, $T$ refers to the number of epochs, $d$ refers to the dimension of representations. $f$ indicates the frequency of permutation matrix $P$ to be updated. For example, if the update frequency is 3, $P$ will be updated every three epochs. $\varepsilon$ refers to the overall privacy budget. $\Omega$ refers to the clipping norm. \textit{SGD} refers to stochastic gradient descent without momentum, i.e., $momentum=0$. We adopt the SGD optimizer in FedOnce-L1 because our analysis of differential privacy is based on SGD.

\begin{table}[htpb]
    \centering
    \small
    \caption{Hyperparameters of FedOnce-L1 on each dataset}
    \label{tab:hyper_l1}
    \setlength\tabcolsep{2pt}
    \begin{tabular}{ccccc ccccc ccc}
    \toprule
        \multirow{2}{*}{\textbf{Dataset}} & \multirow{2}{*}{$\varepsilon$} & \multicolumn{4}{c}{\textbf{Guest Model}} & \multicolumn{4}{c}{\textbf{Host Model}} & \multirow{2}{*}{$d$} & \multirow{2}{*}{$f$} & \multirow{2}{*}{$\Omega$} \\ \cmidrule(lr){3-6}\cmidrule(lr){7-10}
         & & $\eta$ & $\lambda$ & $b$ & $T$ & $\eta$ &  $\lambda$ & $b$ & $T$ \\
         \midrule
         \multirow{4}{*}{gisette} & 2 & 0.3 & 0 & 32 & 10 & 0.3 & 0 & 32 & 10 & 3 & 1 & 1.0  \\
            & 4 & 0.2 & 0 & 32 & 6 & 0.2 & 0 & 32 & 30 & 3 & 1 & 1.0  \\
            & 6 & 0.3 & 0 & 128 & 10 & 0.6 & 0 & 128 & 15 & 3 & 1 & 1.5  \\
            & 8 & 0.3 & 0 & 128 & 10 & 0.6 & 0 & 128 & 40 & 3 & 1 & 1.5  \\
        \midrule
        \multirow{4}{*}{phishing} & 2 & 0.3 & 0 & 32 & 10 & 0.3 & 0 & 32 & 30 & 3 & 1 & 1.0  \\
            & 4 & 0.2 & 0 & 32 & 10 & 0.2 & 0 & 32 & 30 & 3 & 1 & 1.0  \\
            & 6 & 0.3 & 0 & 32 & 10 & 0.1 & 0 & 32 & 30 & 3 & 1 & 1.0  \\
            & 8 & 0.3 & 0 & 128 & 10 & 0.3 & 0 & 128 & 40 & 3 & 1 & 1.5  \\
        \midrule
        
        \multirow{4}{*}{\makecell{MNIST\\Fashion-MNIST}} & 2 & 0.5 & 0 & 256 & 6 & 0.8 & 0 & 256 & 15 & 8 & 1 & 1.5  \\
            & 4 & 0.5 & 0 & 256 & 6 & 0.9 & 0 & 256 & 15 & 8 & 1 & 1.5  \\
            & 6 & 0.5 & 0 & 256 & 7 & 0.9 & 0 & 256 & 15 & 8 & 1 & 1.5  \\
            & 8 & 0.5 & 0 & 512 & 10 & 0.9 & 0 & 512 & 30 & 8 & 1 & 1.5  \\
        \midrule
        
        \multirow{4}{*}{\makecell{UJIIndoorLoc\\Superconduct}} & 2 & 0.4 & 0 & 128 & 10 & 0.5 & 0 & 128 & 10 & 6 & 1 & 1.5  \\
            & 4 & 0.4 & 0 & 128 & 12 & 0.5 & 0 & 128 & 20 & 6 & 1 & 1.5  \\
            & 6 & 0.4 & 0 & 128 & 12 & 0.5 & 0 & 128 & 30 & 6 & 1 & 1.5 \\
            & 8 & 0.4 & 0 & 128 & 12 & 0.3 & 0 & 128 & 60 & 6 & 1 & 1.5 \\
            
        \bottomrule
    \end{tabular}
\end{table}

\section{Supplement Experiment on CIFAR-10}
\wzm{To evaluate the performance of FedOnce on complex image datasets, additional experiments are conducted on CIFAR-10 which consists of 60000 32x32 color images in 10 classes. We divide the features both horizontally and vertically into four parties; each party holds a 16x16 partial image with full three channels. Since it is an open problem to fit more advanced deep learning models in unsupervised learning contexts, we use CNN0 in Fig. 4 (the same model we use for other image datasets) as a starting point to demonstrate our proposal on more complex datasets. The experimental setting is similar to that in Section~5.3. We present the performance of each approach as the communication cost increases in Fig.~\ref{fig:comm_cifar10}.}

\wzm{As can be observed from Fig.~\ref{fig:comm_cifar10}, multi-round FedOnce 1) significantly outperforms SplitNN when the communication size is small and 2) consistently outperforms SplitNN as the communication size increases.}

\begin{figure}[t!]
    \centering
    \includegraphics[width=.75\linewidth]{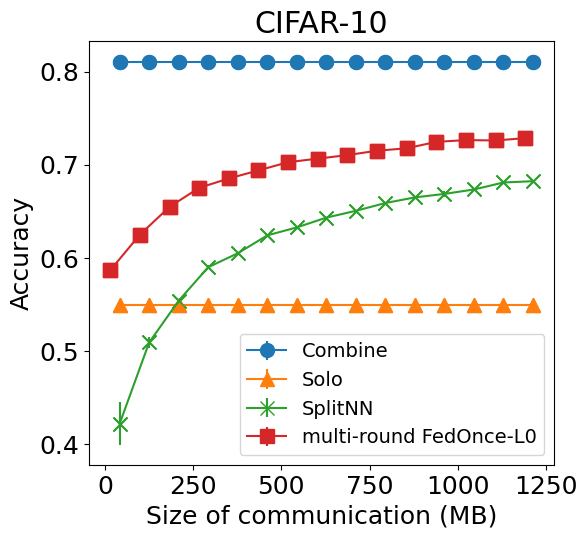}
    \caption{\wzm{Performance of multi-round FedOnce-L0 and FedOnce-L0 (i.e., the first data point of multi-round FedOnce-L0) on CIFAR-10 ($k=4$)}}
    \label{fig:comm_cifar10}
\end{figure}

\section{Discussion of Laggy Parties}
\wzm{During the distributed machine learning process, some parties may respond slowly due to communication bottlenecks. The existence of these parties, denoted as \textit{laggy parties}, is one of the major hurdles of federated learning. Although FedOnce significantly reduces the effect of laggy parties by communicating for only one round, the impact of laggy parties has not been eliminated. Therefore, in this section, we present an intuitive solution to tackle this issue in FedOnce.}

\wzm{Setting a threshold $t_r$ for the responding time, all the parties that respond later than $t_r$ are regarded as laggy parties which do not participate in the training of FedOnce. Specifically, all the representations from laggy parties are set to zero. Thus, FedOnce can still learn from the host party and the representations from other parties.}

\end{document}